\pdfoutput=1
\documentclass[11pt]{article}
\usepackage{subfig}
\usepackage{amssymb,amsmath,amsthm}
\usepackage{graphics,epsfig}
\usepackage{algorithmic}
\usepackage{algorithm}
\usepackage[margin=1 in]{geometry}
\usepackage{mathtools}
\usepackage{cuted}
\usepackage{hhline}
\usepackage{multirow}
\usepackage{array}
\usepackage{tabularx}
\usepackage{xcolor}
\usepackage{tikz}
\usetikzlibrary{arrows.meta,positioning,calc}
\usepackage{cite}
\usepackage{url}
\newtheorem{defs}{Definition}
\newtheorem{lem}{Lemma}[section]

\newtheorem{thm}[lem]{Theorem}

\newtheorem{fact}[lem]{Fact}

\newenvironment{restated}[2]{\par\smallskip\noindent\textbf{#1~\ref{#2}} (restated)\textbf{.}\ \itshape}{\par\smallskip}
\makeatletter
\providecommand{\bstctlcite}[1]{\@bsphack\@for\@citeb:=#1\do{\edef\@citeb{\expandafter\@firstofone\@citeb}\if@filesw\immediate\write\@auxout{\string\citation{\@citeb}}\fi}\@esphack}
\makeatother
\title{Approximation Algorithms for Scheduling Unrelated Crowd Workers with Contact Delays in Mobile Social Networks\thanks{This work has been submitted to the IEEE for possible publication. Copyright may be transferred without notice, after which this version may no longer be accessible.}}
\author{Chi-Yeh~Chen 
\\ Department of Computer Science and Information
Engineering, \\ National Cheng Kung University, \\
Taiwan, ROC. \\
chency@mail.csie.ncku.edu.tw.}

\begin{document}
\bstctlcite{IEEEexample:BSTcontrol}
\maketitle
\begin{abstract}
This paper studies a distributed task scheduling problem in which a requester assigns independent tasks to heterogeneous mobile devices (crowd workers) that it reaches only through opportunistic contacts in a mobile social network. The problem generalizes scheduling on unrelated parallel machines: the service time of a task depends on the crowd worker that performs it, every task incurs a worker-dependent contact delay, and the objective is to minimize the total weighted completion time. For identical crowd workers, we tighten the approximation ratio of the Largest-Ratio-First (LRF) algorithm and update the competitive ratio of its online version. For unrelated crowd workers, we show that the time-indexed rounding framework of Im and Li remains valid when worker-dependent contact delays are added, which yields a randomized algorithm with an expected approximation ratio of 1.45; this improves upon the expected ratio of $1.5+\epsilon$ obtained in our previous work for identical crowd workers. We also derandomize this algorithm. To improve computational efficiency, we further propose an algorithm that exploits the optimality of the Smith-ratio order on a single crowd worker, and we prove that it is a deterministic 1.5-approximation algorithm and that its polynomial-time version achieves an approximation ratio arbitrarily close to 1.5. Experiments on synthetic instances and on instances derived from real contact traces show that this algorithm outperforms three variants of LRF in all tested settings, with the largest advantage when the crowd workers are highly heterogeneous, performs on par with a local-search heuristic for unrelated machines, and is within 0.4\% of the optimum on average on small instances that can be solved exactly.
\end{abstract}

\section{Introduction}\label{sec:introduction}
Mobile devices equipped with rich sensors have become ubiquitous, and their embedded sensors are now used in applications such as air-quality monitoring~\cite{Dutta2009}, urban noise mapping~\cite{Rana2010}, road-traffic monitoring~\cite{Farkas2014}, and road-surface assessment~\cite{Singh2017}. Such applications typically generate a large number of sensing tasks that must be processed promptly with limited resources. Because the computing and sensing capabilities of a single mobile user are rarely sufficient for a large project consisting of many small, independent tasks, it is natural to distribute these tasks among other users; this idea underlies \textit{mobile crowdsourcing}.

Operating a large-scale crowdsourcing platform entails considerable management and maintenance costs. If a requester instead distributes tasks within its own social circle, the existing Mobile Social Network (MSN) can be used to carry them out, which reduces both overhead and cost~\cite{Xiao2015, Zhang2025}. We denote the requester by $r$. Because crowd workers differ in their sensing equipment and behavior, their performance on the same task may differ~\cite{Fan2015, Wang2025, 9420265}. When $r$ assigns task $j$ to crowd worker $i$ ($1 \leq i \leq m$), worker $i$ reports the result to $r$ after completing the task. The time needed to complete a task therefore consists of the meeting times between $r$ and $i$ before and after the task and the service time $p_{ij}$ of the task itself.

This paper studies task assignment and scheduling with unrelated (heterogeneous) crowd workers in MSN-based mobile crowdsourcing, with the objective of minimizing the total weighted completion time when service times depend on the assigned crowd worker. Communication between the requester and the crowd workers is bidirectional: following~\cite{Zhang2025, chen2025approx}, each task requires one interaction for its distribution and another for the feedback of its result. We study both the offline and the online setting. Fig.~\ref{fig:system} illustrates the system model. From a scheduling perspective, the crowd workers form a heterogeneous distributed computing system in which each crowd worker is a processor with task-dependent processing times, as in heterogeneous computing systems~\cite{ali2000task, braun2001comparison, topcuoglu2002heft}, and the opportunistic contacts add a processor-dependent delay to every task. The problem therefore generalizes the classical problem of minimizing the total weighted completion time on unrelated parallel machines, and our algorithms and analyses apply to any distributed system in which each task incurs an additive delay that depends only on the processor to which it is dispatched.

\begin{figure}[!t]
\centering
\begin{tikzpicture}[font=\footnotesize,
  user/.style={circle,draw,thick,minimum size=7mm,inner sep=0pt},
  req/.style={circle,draw,very thick,fill=gray!20,minimum size=9mm,inner sep=0pt},
  >={Stealth[length=2mm]}]
  \node[req] (r) at (0,0) {$r$};
  \node[user] (w1) at (-2.6,1.3) {$1$};
  \node[user] (w2) at (-2.6,-0.1) {$2$};
  \node[user] (wm) at (-2.6,-1.5) {$m$};
  \node at (-2.6,-0.8) {$\vdots$};
  \foreach \w/\lab in {w1/{\phi_{1}},w2/{\phi_{2}},wm/{\phi_{m}}}{
    \draw[->,thick] (r) to[bend left=12] node[pos=0.55,fill=white,inner sep=1pt]{$\lab$} (\w);
    \draw[->,dashed] (\w) to[bend left=12] (r);
  }
  \node[align=center] at (-0.2,-1.95) {requester $r$ with task set $\mathcal{J}$};
  \begin{scope}[shift={(0.9,0.9)}]
    \node[anchor=west] at (0,1.25) {timeline of crowd worker $i$:};
    \draw[fill=gray!25] (0,0) rectangle (0.9,0.35) node[midway]{$\phi_i$};
    \draw[fill=white] (0.9,0) rectangle (1.6,0.35) node[midway]{$p_{ij_1}$};
    \draw[fill=white] (1.6,0) rectangle (2.5,0.35) node[midway]{$p_{ij_2}$};
    \draw[fill=white] (2.5,0) rectangle (3.0,0.35) node[midway]{$\cdots$};
    \draw[->,dashed] (1.6,0.35) -- (1.95,0.65) node[above,inner sep=1pt]{\scriptsize $+\phi_i$};
    \draw[->,dashed] (2.5,0.35) -- (2.85,0.65) node[above,inner sep=1pt]{\scriptsize $+\phi_i$};
    \draw[->] (0,-0.1) -- (4.1,-0.1) node[right]{$t$};
    \node[anchor=north,align=center,text width=4cm] at (1.95,-0.15) {distribution (EMT), sequential processing in Smith-ratio order; each result is returned after an EMT};
  \end{scope}
  \node[anchor=west,align=left] at (0.9,-1.3) {$\rightarrow$ task distribution\\ $\dashrightarrow$ result feedback};
\end{tikzpicture}
\caption{System model. The requester $r$ distributes its tasks to the crowd workers when they meet, each crowd worker $i$ processes its assigned tasks sequentially with worker-dependent service times $p_{ij}$, and the result of each task is returned as soon as the task is finished, after which the next task is processed. Each task incurs the expected meeting time $\phi_{i}$ once for distribution and once for the feedback of its result.}
\label{fig:system}
\end{figure}

\subsection{Related Work}
Task scheduling has been studied extensively~\cite{Allahverdi2008, Ding2016, Bridi2016, Bhatti2021}. In the MSN setting, Xiao \textit{et al.}~\cite{Xiao2015} studied crowdsensing scheduling, in which mobile devices collect environmental data, and proposed greedy offline and online algorithms that reduce the average completion time. Zhang \textit{et al.}~\cite{Zhang2025} proposed algorithms for minimizing the total weighted completion time and the makespan. Chen~\cite{chen2025approx} identified a gap in the approximation analysis of~\cite{Zhang2025} and established the corrected bound $\max\{\frac{3}{2},\frac{w_{max}\phi_{max}}{w_{min}\phi_{min}}\}$ for LRF, where $\phi_{max}$ and $\phi_{min}$ are the maximum and minimum expected inter-contact times and $w_{\max}$ and $w_{\min}$ are the maximum and minimum task weights; Chen also gave a deterministic algorithm with approximation ratio $\max\{2.5,1+\epsilon\}$ for any $\epsilon>0$.

Task scheduling in mobile crowdsourcing is closely related to parallel machine scheduling~\cite{Maiti2020}, which asks how to assign jobs to several machines so as to optimize objectives such as the makespan or the total weighted completion time. The main difference is that crowd workers are mobile and their availability changes over time, whereas machines are static and always reachable. When communication delays are ignored, the following classical results are known:
\begin{itemize}
\item \textbf{Minimizing the makespan:} Graham~\cite{graham1969bounds} showed that the Longest Processing Time (LPT) rule has an approximation ratio of $\frac{4}{3} - \frac{1}{3m}$ on identical machines.

\item \textbf{Minimizing the total weighted completion time:} Eastman \textit{et al.}~\cite{eastman1964} showed that the Largest-Ratio-First (LRF) rule is a 1.5-approximation on identical machines, and Kawaguchi and Kyan~\cite{Kawaguchi1986} improved the bound to $\frac{\sqrt{2}+1}{2} \approx 1.207$.
\end{itemize}

For minimizing the total weighted completion time on unrelated machines, Bansal \textit{et al.}~\cite{bansal2016lift} were the first to break the long-standing approximation barrier of 1.5. Li~\cite{li2020scheduling} improved the ratio to $1.5 - \frac{1}{6000}$, Im and Shadloo~\cite{im2020weighted} obtained 1.488 by iterative fair contention resolution, and Im and Li~\cite{Im2023} achieved 1.45 with a rounding scheme that allows non-disjoint groups of jobs.

Beyond task scheduling, research on mobile crowdsourcing covers incentive mechanisms~\cite{Fan2015}, quality control~\cite{yan2015}, and security and privacy~\cite{Wu2014}. In event-based social networks, She \textit{et al.}~\cite{She2015} introduced the utility-aware social event-participant planning problem, which arranges personalized event schedules for participants, and related variants were studied in~\cite{Cheng2017, Cheng2021, Liu2012}; these works optimize itineraries for individual users rather than task allocation. For utility maximization, Guo \textit{et al.}~\cite{Guo2024TMC} studied the multi-task diffusion maximization problem, which maximizes the total utility of several crowdsourcing tasks under a budget, and Wang \textit{et al.}~\cite{Wang2025} designed an online worker scheduling mechanism for platform-centric spatiotemporal crowdsourcing that maximizes the long-term utility of the platform subject to worker-side and platform-side constraints.

\subsection{Our Contributions}
The contributions of this paper are summarized below. For each result we give the corresponding theorem, and for results that build on existing work we indicate the source.

\begin{itemize}
\item \textbf{Tighter analysis of LRF (Theorem~\ref{thm:thm2}, refined analysis of an existing algorithm).} For identical crowd workers, we improve the approximation ratio of the LRF algorithm from $\max\left\{\frac{3}{2},\frac{w_{max}\phi_{max}}{w_{min}\phi_{min}}\right\}$~\cite{chen2025approx} to $\max\left\{\frac{3}{2},\frac{\phi_{max}}{\phi_{min}}\right\}$. The improvement comes from bounding the contact delay of each task separately from its processing delay (Lemmas~\ref{lem:lem1} and~\ref{lem:lem2} in the appendix), so that the task weights cancel. We also update the analysis of the online version and obtain a competitive ratio of $\alpha\left(1 + \frac{\Phi_{max} \sum_{j=1}^{n}w_{j}}{\sum_{j=1}^{n}w_{j}(\Phi_{min}+p_{j})}\right)$ (Theorem~\ref{thm:thm3}), where $\alpha$ is the approximation ratio of the offline algorithm used.

\item \textbf{1.45-approximation for unrelated crowd workers (Theorem~\ref{thm:thm44}, adaptation of Im and Li~\cite{Im2023}).} For unrelated crowd workers, we extend the time-indexed formulation with worker-dependent contact delays and show that the rounding scheme of Im and Li~\cite{Im2023} still yields an expected approximation ratio of 1.45. The ratio itself is inherited from~\cite{Im2023}; our contribution is the formulation and the proof that the additional per-worker delay does not invalidate any step of their analysis. This improves the expected ratio of $1.5+\epsilon$ of the randomized RIS algorithm in~\cite{chen2025approx}, which was designed for identical crowd workers, a special case of the problem considered here. We also present a derandomization of this algorithm (Section~\ref{sec:Algorithm4}).

\item \textbf{Efficient deterministic algorithm with a provable guarantee (Theorem~\ref{thm:EDTS}).} To improve efficiency, we propose an algorithm that exploits the fact that an optimal schedule on a single crowd worker processes tasks in non-increasing order of their Smith ratios~\cite{smith1956}. We prove that it is a deterministic 1.5-approximation algorithm, and that its polynomial-time version based on an interval-indexed LP is a deterministic $(1.5+\epsilon/2)$-approximation algorithm for any fixed $\epsilon>0$ (Theorem~\ref{thm:EDTS-interval}); the latter is the version evaluated in our experiments. Exploiting the Smith-ratio order reduces the running time of the derandomized assignment step to $O(mn^2)$, excluding the time to solve the LP. In comparison, the deterministic DIS algorithm of our previous work~\cite{chen2025approx}, designed for identical crowd workers, guarantees $\max\{2.5,1+\epsilon\}$ in general and $1.5+\epsilon$ only when the contact times or service times are sufficiently large, and its assignment step takes $O(n^2m^2L^2)$ time; the polynomial-time EDTS guarantees $1.5+\epsilon/2$ for every instance of the more general unrelated problem. In our experiments, it outperforms three variants of LRF in all tested settings, reducing the WCTR by 3.5\% to 20\% compared with the best variant; it performs on par with a local-search heuristic for unrelated machines while, unlike that heuristic, providing a worst-case guarantee; and it is within 0.4\% of the optimum on average on small instances solved exactly.
\end{itemize}

Table~\ref{tab:comparison} compares this paper with our previous work~\cite{chen2025approx} and with the most closely related studies.

\begin{table*}[!t]
\caption{Comparison with Closely Related Work}
\label{tab:comparison}
\centering
\footnotesize
\renewcommand{\arraystretch}{1.15}
\setlength{\tabcolsep}{4pt}
\begin{tabular}{|>{\raggedright\arraybackslash}p{1.15in}|>{\raggedright\arraybackslash}p{0.95in}|>{\raggedright\arraybackslash}p{1.45in}|>{\raggedright\arraybackslash}p{0.9in}|>{\raggedright\arraybackslash}p{1.45in}|}
\hline
 & Zhang \textit{et al.}~\cite{Zhang2025} & Chen~\cite{chen2025approx} & Im and Li~\cite{Im2023} & This paper \\
\hline
Crowd worker model & Identical & Identical & Unrelated machines & Unrelated (identical as a special case) \\
\hline
Contact delay between requester and workers & Yes & Yes & No & Yes \\
\hline
LRF with identical workers & $\frac{3}{2}$ claimed (proof gap identified in~\cite{chen2025approx}) & $\max\{\frac{3}{2},\frac{w_{max}\phi_{max}}{w_{min}\phi_{min}}\}$; $2-\frac{1}{m}$ if $p_{min}>\Phi_{min}$ & -- & $\max\{\frac{3}{2},\frac{\phi_{max}}{\phi_{min}}\}$ \\
\hline
LP relaxation & -- & Interval-indexed (polynomial size) & Time-indexed (pseudo-polynomial size) & Time-indexed (RTS, DTS, EDTS); interval-indexed (polynomial-time EDTS) \\
\hline
Randomized LP-based algorithm (expected ratio) & -- & RIS: $1.5+\epsilon$ & $1.45$ (no contact delay) & RTS: $1.45$ \\
\hline
Deterministic LP-based algorithm & -- & DIS: $\max\{2.5,1+\epsilon\}$; $1.5+\epsilon$ only if the contact and service times are sufficiently large; assignment step $O(n^2m^2L^2)$ & -- & EDTS: $1.5$; polynomial-time EDTS: $1.5+\frac{\epsilon}{2}$ for every instance; assignment step $O(mn^2)$ \\
\hline
Online setting & CosMOS & CosMOS (corrected bound); ODIS (LP-based) & -- & CosMOS with a tighter lower bound on OPT \\
\hline
\end{tabular}
\end{table*}

\subsection{Organization}
The rest of this paper is organized as follows. Section~\ref{sec:Preliminaries} introduces the model and notation. Section~\ref{sec:LRF} analyzes the LRF algorithm and its online version. Section~\ref{sec:Algorithm3} presents the randomized algorithm based on a time-indexed LP relaxation, and Sections~\ref{sec:Algorithm4} and~\ref{sec:Algorithm5} present deterministic algorithms, the latter with improved efficiency, a 1.5-approximation guarantee, and a polynomial-time version. Section~\ref{sec:Results} reports the experimental results, and Section~\ref{sec:Conclusion} concludes the paper. The appendices contain the proofs omitted from the main text, the details of the derandomized algorithm of Section~\ref{sec:Algorithm4}, additional experimental results, and a table of notation.

\section{Notation and Preliminaries}\label{sec:Preliminaries}
We adopt the system model of~\cite{Zhang2025, chen2025approx}. A mobile social network contains a requester $r$, i.e., the user who wants to outsource tasks, and a set $\mathcal{M}$ of $m$ other users who may perform them, called crowd workers. Two users can communicate only when their devices are within communication range, and a contact is assumed to last long enough to exchange a task or a result. Crowd workers do not communicate with each other, so only the contacts between $r$ and each crowd worker matter. These contacts are modeled stochastically: let $\phi_{i}$ denote the expected inter-contact time between $r$ and crowd worker $i$. The distribution of the inter-contact times may be arbitrary; a common choice is the exponential distribution~\cite{Gao2009}.


The requester has a set $\mathcal{J}$ of $n$ indivisible tasks. Each task $j \in \mathcal{J}$ has a weight $w_j \geq 0$, and its Required Service Time (RST) on crowd worker $i$ is $p_{ij} \geq 0$. Each task must be assigned to exactly one crowd worker, whereas a crowd worker may receive several tasks. In addition to its service time, every task incurs a communication overhead for its distribution and for the feedback of its result.

A schedule partitions the tasks into $m$ disjoint sets $\{\mathcal{J}_{1}, \mathcal{J}_{2}, \ldots, \mathcal{J}_{m}\}$, where $\mathcal{J}_{i}$ is the set of tasks assigned to crowd worker $i$, and each crowd worker processes its tasks sequentially. The completion time of a task $j \in \mathcal{J}_{i}$ consists of three parts~\cite{Zhang2025}: (a) the time until $r$ meets crowd worker $i$ to distribute the task; (b) the processing time of $j$ and of all tasks in $\mathcal{J}_{i}$ processed before it; and (c) the time until $r$ and crowd worker $i$ meet again so that the result can be returned.

Since the meeting times are random, we use the Expected Meeting Time (EMT) of Zhang \textit{et al.}~\cite{Zhang2025}: the EMT of crowd worker $i$ is $\phi_{i}$, and it is used for both part (a) and part (c). Hence the completion time $C_{j}$ of task $j\in\mathcal{J}_{i}$ is the sum of two EMTs and the total processing time of $j$ and its predecessors in $\mathcal{J}_{i}$. The objective is to minimize the total weighted completion time $\sum_{j=1}^{n} w_{j}C_{j}$.

We make the scope of this communication model explicit. As in~\cite{Zhang2025, chen2025approx}, the requester sends the tasks in $\mathcal{J}_{i}$ to crowd worker $i$ one after another, crowd worker $i$ processes them sequentially, and the result of each task is returned as soon as the task is finished, after which the next task is processed. The distribution delay before processing and the feedback delay after a task are each $\phi_{i}$ in expectation. Hence every task assigned to crowd worker $i$ incurs the same expected contact delay $\Phi_{i}=2\phi_{i}$, independently of its position in the processing sequence. The model uses expected inter-contact times; the variability and temporal dependence of the realized contacts, as well as missed contacts, are not modeled explicitly in the offline setting. The online algorithm in Section~\ref{sec:LRF} reacts to the realized meeting times, and its competitive analysis (Theorem~\ref{thm:thm3}) quantifies the loss caused by the gap between expected and realized contacts. Extending the model to stochastic contact processes is left for future work.

Following~\cite{Zhang2025}, we use the expected workload of a crowd worker.
\begin{defs} (Expected Workload (EW)~\cite{Zhang2025}). The expected workload $EW_{i}$ of crowd worker $i$ is the sum of the EMT for distributing tasks to $i$, the total service time of the tasks assigned to $i$, and the EMT for returning the results to the requester. If crowd worker $i$ is assigned the set $\mathcal{J}_{i}$, then $$EW_{i} = 2\phi_{i} + \sum_{j \in \mathcal{J}_{i}} p_{ij},$$ and $EW_{i} = 2\phi_{i}$ if $\mathcal{J}_{i} = \emptyset$.
\end{defs}

Note that $C_{j} = EW_{i}$ if $j$ is the last task processed by crowd worker $i$. We call $\Phi_{i}=2\phi_{i}$ the total contact time of crowd worker $i$, and define $\phi_{max}=\max_{i\in [1, m]} \phi_{i}$, $\phi_{min}=\min_{i\in [1, m]} \phi_{i}$, $\Phi_{max}=2\phi_{max}$, and $\Phi_{min}=2\phi_{min}$. When the crowd workers are identical (Section~\ref{sec:LRF}), the service time of a task does not depend on the crowd worker, i.e., $p_{ij}=p_{j}$ for all $i\in\mathcal{M}$, where $p_{j}$ denotes the service time of task $j$. Table~\ref{tab:notations} in the appendix summarizes the notation.


\section{Largest-Ratio-First Based Task Scheduling}\label{sec:LRF}
The LRF algorithm for this problem was proposed by Zhang \textit{et al.}~\cite{Zhang2025}. It sorts the tasks in non-increasing order of their Smith ratios~\cite{smith1956} and assigns each task in turn to the crowd worker with the smallest current expected workload. In~\cite{chen2025approx}, we showed that its approximation ratio is $\max\left\{\frac{3}{2},\frac{w_{max}\phi_{max}}{w_{min}\phi_{min}}\right\}$, where $w_{\max}$ and $w_{\min}$ are the maximum and minimum task weights. That analysis used relatively loose upper and lower bounds on the costs of LRF and of an optimal solution, respectively; here we derive tighter bounds and thus a better ratio. Since LRF is designed for identical crowd workers, we set $p_{ij}=p_{j}$ for all $i\in \mathcal{M}$ and $j\in \mathcal{J}$ in this section. Its pseudocode is given in Algorithm~\ref{Alg_LRF} in the appendix. Sorting takes $O(n \log n)$ time and the assignment takes $O(nm)$ time, so the running time is $O(nm + n \log n)$.

We bound the cost of an optimal schedule from below by the bound of Eastman \textit{et al.}~\cite{eastman1964} plus the minimum contact delay of every task, and the cost of LRF from above by the average expected workload at the time each task is assigned plus the maximum contact delay. Because the contact delays and the processing delays are bounded separately, the task weights cancel and we obtain the following improved approximation ratio; the proof is given in Appendix~\ref{supp:LRF}.
\begin{thm}\label{thm:thm2}
For identical crowd workers, the LRF algorithm satisfies
\begin{eqnarray*}\label{ineq:02}
\frac{WCT_{LRF}}{WCT_{OPT}} \leq \max\left\{\frac{3}{2},\frac{\phi_{max}}{\phi_{min}}\right\}.
\end{eqnarray*}
\end{thm}

\subsection{Online Task Scheduling}
In~\cite{Zhang2025, chen2025approx}, an offline algorithm is turned into an online one as follows. Whenever the requester meets a crowd worker $i \in \mathcal{M}$, the algorithm schedules all remaining tasks on the crowd workers that have not yet received tasks and hands over to $i$ the tasks that this schedule assigns to $i$. In this way, any offline algorithm can be converted into an online algorithm. The Cost Minimized Online Scheduling (CosMOS) algorithm~\cite{Zhang2025} implements this approach: when the requester meets crowd worker $i$ at time $t_{i}$, it applies an offline algorithm $A$, such as LRF, to the remaining tasks and crowd workers to determine the tasks of crowd worker $i$, which $i$ processes in Smith-ratio order, and then removes these tasks and crowd worker $i$. Its pseudocode is given in Algorithm~\ref{Alg_ONLRF} in the appendix; its running time depends on the offline algorithm $A$.

The next theorem refines the corresponding result of~\cite{chen2025approx} by using a better lower bound on the optimal cost. It gives two bounds on the competitive ratio: the first is tighter, and the second is weaker but involves fewer parameters.
\begin{thm}\label{thm:thm3}
Let $\Lambda_{OPT} = \{\mathcal{J}_1^*, \mathcal{J}_2^*, \dots, \mathcal{J}_m^* \}$ be an optimal schedule computed by an omniscient scheduler that knows in advance all meeting times between the requester and the crowd workers, and let $WCT_{OPT}$ be its cost. If the offline algorithm $A$ has approximation ratio $\alpha$, then
\begin{eqnarray*}\label{ineq:11}
\frac{WCT_{CosMOS}}{WCT_{OPT}} & \leq & \alpha\left(1 + \frac{\Phi_{max} \sum_{j=1}^{n}w_{j}}{\sum_{j=1}^{n}w_{j}(\Phi_{min}+p_{j})}\right) \\
                               & \leq & \alpha\left(1 + \frac{\Phi_{max}}{\Phi_{min}+p_{min}}\right)
\end{eqnarray*}
where $p_{min}$ is the minimum required service time.
\end{thm}

The proof is given in Appendix~\ref{supp:LRF}.

\section{Randomized Approximation Algorithm with Time-indexed Approach}\label{sec:Algorithm3}
This section presents a randomized approximation algorithm for unrelated crowd workers; Section~\ref{sec:Algorithm4} derandomizes it by assigning the tasks one by one so as to minimize the expected total weighted completion time. The algorithm builds on the work of Im and Li~\cite{Im2023}. We modify the time-indexed integer programming (IP) formulation used in~\cite{im2016better, li2020scheduling, im2020weighted, Im2023} by adding the total contact time of each crowd worker to the completion times of its tasks. The binary variable $x_{ijs}$ equals 1 if task $j$ starts on crowd worker $i$ at time $s$. Since preemption is not allowed, task $j$ then completes at time $\Phi_{i} + s + p_{ij}$. Let $T$ be an upper bound on the number of time steps. As in~\cite{im2016better}, we assume that $T$ is polynomially bounded in the input size; this is without loss of generality at the cost of a factor of $1+\epsilon$ in the approximation ratio. The modified time-indexed IP is as follows.

\begin{subequations}\label{LP:1}
\begin{align}
& \text{min}  && \sum_{j \in \mathcal{J}} w_{j} \sum_{i \in \mathcal{M}, s\in [0, T)} x_{ijs}(\Phi_{i}+s+p_{ij})     &   & \tag{\ref{LP:1}} \\
& \text{s.t.} && \sum_{i \in \mathcal{M}, s\in [0, T)} x_{ijs} = 1, && \forall j\in \mathcal{J} \label{LP1:a} \\
&  && \sum_{j \in \mathcal{J}, s\in [t-p_{ij}, t)} x_{ijs} \leq 1, && \forall i, t\in [T] \label{LP1:b} \\
&  && x_{ijs} =0, && \forall i, j, s>T-p_{ij} \label{LP1:c} \\
&  && x_{ijs} \in \left\{0, 1\right\}, && \forall i, j, s\in[0, T) \label{LP1:d}
\end{align}
\end{subequations}

Constraint~(\ref{LP1:a}) ensures that every task is scheduled, constraint~(\ref{LP1:b}) ensures that each crowd worker processes at most one task at a time, and constraint~(\ref{LP1:c}) ensures that all tasks finish by time $T$. Replacing the integrality constraint by $x_{ijs} \geq 0$ yields the LP relaxation of~(\ref{LP:1}); let $x$ be an optimal solution. Following~\cite{im2020weighted, Im2023}, each $x_{ijs} > 0$ is viewed as a rectangle $R_{ijs}$ of height $x_{ijs}$ spanning the interval $[s, s + p_{ij}]$; in completion times, the rectangles of crowd worker $i$ are shifted by its total contact time $\Phi_{i}$.

As shown by Li~\cite{li2020scheduling} for the problem without contact delays, independent rounding of this LP yields a 1.5-approximation. Each task $j$ independently selects a rectangle $R_{ijs}$ with probability $x_{ijs}$ and is assigned to crowd worker $i$; it then draws a random offset $\tau_j$ uniformly from $[0, p_{ij}]$ and sets $\theta_j = s + \tau_j$. The tasks assigned to each crowd worker are processed in non-decreasing order of $\theta$.

The expected completion time of task $j$ is determined by the expected total processing time of the tasks on the same crowd worker with smaller $\theta$ values; by linearity of expectation, it is the sum of the expected delays caused by the individual tasks $j'$.

Because the offsets are uniform, for a fixed $\theta_j$ and a task $j'$ that selects rectangle $R_{ij's'}$, the probability that $\theta_{j'} < \theta_j$ is $\min\{1,\max(0, \theta_j - s') / p_{ij'}\}$~\cite{im2020weighted}, i.e., the length of the part of $R_{ij's'}$ before $\theta_j$ divided by $p_{ij'}$. Hence the expected delay that $R_{ij's'}$ causes to task $j$ equals the area of the part of $R_{ij's'}$ that lies before $\theta_j$.

By constraint~(\ref{LP1:b}), the total area of the rectangles on crowd worker $i$ before time $\theta_j$ is at most $\theta_j$. Therefore, $$E[C_j \mid \theta_j, R_{ijs}] \leq \Phi_i + \theta_j + p_{ij}.$$ Since $E[\theta_j \mid R_{ijs}] = s + p_{ij}/2$, we have $E[C_j \mid R_{ijs}] \leq \Phi_i + s + 1.5p_{ij}$, and removing the conditioning gives $$E[C_j] \leq \sum_{i, s} x_{ijs} (\Phi_i + s + 1.5p_{ij}).$$ By linearity of expectation, the expected objective is at most 1.5 times the LP value.

Independent rounding cannot do better than 1.5, and breaking this barrier was a long-standing open problem~\cite{chekuri2004approximation, kumar2008minimum, schulz2002scheduling, schuurman1999polynomial, sviridenko2013approximating}. It was first resolved by Bansal \textit{et al.}~\cite{bansal2016lift}, who introduced a rounding scheme with strong negative correlation within groups of jobs. Since the total weighted completion time depends on the delays that tasks on the same crowd worker impose on each other, rounding should avoid placing tasks together whenever possible. Ideally, the probability that tasks $j$ and $j'$ are both assigned to crowd worker $i$ would be strictly less than $x_{ij} x_{ij'}$ for every pair; this cannot hold for all pairs simultaneously, but Bansal \textit{et al.} showed that it can be achieved within disjoint groups of tasks on each crowd worker.

Im and Li~\cite{Im2023} introduced a Strong Negative Correlation (SNC) scheme. Given a vector $y \in [0, 1]^{\mathcal{M} \times \mathcal{J}}$ with $\sum_{i \in \mathcal{M}} y_{ij} = 1$ for every task $j \in \mathcal{J}$, the edges incident to each crowd worker $i$ are partitioned into groups whose total $y$-value is at most 1. An SNC rounding scheme randomly assigns each task to a crowd worker so that each task is assigned to crowd worker $i$ with probability exactly $y_{ij}$, two distinct tasks are non-positively correlated with respect to being assigned to the same crowd worker, and strongly negatively correlated with respect to being assigned to the same group (see Theorem~\ref{thm:4} below for the precise statement). In contrast to the scheme of Bansal \textit{et al.}, the partition in the scheme of Im and Li is fractional: a task may belong to several groups of the same crowd worker, i.e., the groups need not be disjoint.

We use the rounding scheme of Im and Li~\cite{Im2023}, which guarantees strong negative correlation between two tasks $j$ and $j'$ in the same group of crowd worker $i$ whenever $i$ dominates neither of them, in the following sense.

\begin{defs}
We say a crowd worker $i\in \mathcal{M}$ dominates a task $j\in \mathcal{J}$ if $\sum_{s}x_{ijs} > 1/2$.
\end{defs}

Formally, there is a set $\mathcal{M}$ of crowd workers, a set $\mathcal{J}$ of tasks, a set $\mathcal{U}$ of groups, a function $g : \mathcal{U}\rightarrow \mathcal{M}$ mapping groups to crowd workers and $\psi: \mathcal{J}\rightarrow \mathcal{U}$ mapping tasks to groups. For each group $u \in \mathcal{U}$ such that $g(u) = i$, and for all $j \in \mathcal{J}$, let
\begin{eqnarray*}
y_{uj} = \sum_{s:R_{ijs}\in u} x_{ijs}
\end{eqnarray*}
be the total height of the rectangles in group $u$ associated with task $j$ (the groups used by our algorithm are constructed in Algorithm~\ref{Alg3} below).

The SNC scheme of Im and Li has the following guarantees.
\begin{thm}[{[\citen{Im2023}, Theorem~2.1]}]\label{thm:4}
Given $\mathcal{M}, \mathcal{J}, \mathcal{U}, g$ and $y$ as described above, there is an efficient randomized algorithm that outputs an assignment $\psi: \mathcal{J}\rightarrow \mathcal{U}$ of tasks to groups satisfying the following properties.
\begin{itemize}
\item (\textbf{Marginal Probabilities}) For every $u\in \mathcal{U}$ and $j\in \mathcal{J}$ we have $Pr[\psi(j) = u] = y_{uj}$.
\item (\textbf{Non-Positive Correlation for a Same Crowd Worker}) For any two distinct tasks $j, j'$ and two (possibly
identical) groups $u, u'\in \mathcal{U}$ with $g(u) = g(u')$, we have $Pr[\psi(j) = u, \psi(j') = u']\leq y_{uj}y_{u'j'}$.
\item (\textbf{Strongly Negative Correlation for a Same Group}) For any two distinct tasks $j, j'\in \mathcal{J}$ and group $u\in \mathcal{U}$ such that $g(u)$ does not dominate any of $j$ and $j'$, we have $Pr [\psi(j) = \psi(j') = u]\leq (1-\eta)y_{uj}y_{uj'}$ , where $\eta = 0.1561$.
\end{itemize}
\end{thm}

\begin{algorithm}
\caption{The RTS Algorithm}
\label{Alg3}
\begin{algorithmic}[1]
    \STATE Solve the LP relaxation of (\ref{LP:1}) to obtain an optimal solution $x$, and let $x_{ij}=\sum_{s}x_{ijs}$. \label{alg3-1}
    \STATE Choose $\rho\in[1,1+\beta)$ such that $\ln\rho$ is uniformly distributed in $[0,\ln(1+\beta))$; for every pair $(i,j)$ with $x_{ij}>0$, choose a shifting parameter $\tau_{ij}$ independently and uniformly from $[0,p_{ij})$. \label{alg3-2}
    \STATE For every crowd worker $i$ and base window $k$, form the group $u_{ik}$ of all rectangles on $i$ belonging to base window $k$; for every pair $(i,j)$, form the group $v_{ij}$ of the rectangles for $j$ on $i$ that belong to no base window. Set $g(u_{ik})=g(v_{ij})=i$ and $y_{uj}=\sum_{s:R_{ijs}\in u}x_{ijs}$. \label{alg3-3}
    \STATE Apply the SNC rounding scheme (Theorem~\ref{thm:4}) to obtain $\psi:\mathcal{J}\rightarrow\mathcal{U}$. For each task $j$ with $\psi(j)=u$, choose its anchor rectangle $R_{ijs}\in u$ with probability $x_{ijs}/y_{uj}$, and set $\sigma(j)\leftarrow i$ and $s_{j}\leftarrow s$. \label{alg3-4}
    \FOR{each task $j \in \mathcal{J}$} \label{alg3-5}
        \IF{crowd worker $\sigma(j)$ does not dominate task $j$}
            \STATE $\theta_j \leftarrow (1+\kappa)s_j + \tau_{\sigma(j)j}$ \label{alg3-6}
        \ELSE
            \STATE $\theta_j \leftarrow (1+\kappa)s_j + \tau_{\sigma(j)j} + 0.2 p_{\sigma(j)j}$ \label{alg3-7}
        \ENDIF
    \ENDFOR
    \FOR{each crowd worker $i \in \mathcal{M}$}
        \STATE Schedule the tasks assigned to $i$ as early as possible in non-decreasing order of $\theta_j$, breaking ties independently at random. \label{alg3-8}
    \ENDFOR
\end{algorithmic}
\end{algorithm}

The \emph{Randomized Time-indexed Scheduling} (RTS) algorithm is detailed in Algorithm~\ref{Alg3}; it follows the rounding algorithm of Im and Li~\cite[Section~3.2]{Im2023} applied to the LP relaxation of~(\ref{LP:1}). Line~\ref{alg3-1} solves the LP. Lines~\ref{alg3-2} and~\ref{alg3-3} construct the groups. The time horizon is divided into base windows by the grid points $\rho(1+\beta)^{k}$, $k\in\mathbb{Z}$, where base window $k$ is the interval $(\rho(1+\beta)^{k-1},\rho(1+\beta)^{k}]$, and a rectangle $R_{ijs}$ belongs to base window $k$ if $s\leq\rho(1+\beta)^{k-1}<s+\tau_{ij}\leq\rho(1+\beta)^{k}$. By constraint~(\ref{LP1:b}), every group has total height at most one, so Theorem~\ref{thm:4} applies. Line~\ref{alg3-4} assigns every task to a crowd worker and a starting time $s_{j}$; by Theorem~\ref{thm:4}, the anchor rectangle of task $j$ is $R_{ijs}$ with probability exactly $x_{ijs}$. Lines~\ref{alg3-5} to~\ref{alg3-7} compute the priorities $\theta_{j}$. Note that the same shifting parameter $\tau_{ij}$ is used both to define the groups and to compute $\theta_{j}$; this coupling is essential for the strong negative correlation to reduce the expected delay in the analysis of~\cite{Im2023}. Finally, line~\ref{alg3-8} sequences the tasks on each crowd worker in non-decreasing order of $\theta_{j}$; ties occur with probability zero. We use the parameters $\kappa=0.3$ and $\beta=12.1$ as in~\cite{Im2023} (we write $\kappa$ instead of the symbol $\alpha$ used in~\cite{Im2023} to avoid a clash with the notation of Section~\ref{sec:LRF}), and $\eta=0.1561$ is the constant of Theorem~\ref{thm:4}.

The following theorem gives the approximation ratio of the RTS algorithm.
\begin{thm}\label{thm:thm44}
In the schedule produced by the RTS algorithm, the expected completion time of each task $j$ is at most $1.45\cdot C_{j}^{*}$, where $C_{j}^{*}=\sum_{i,s}x_{ijs}(\Phi_{i}+s+p_{ij})$ is the fractional completion time of $j$ in the optimal solution $x$ of the LP relaxation of~(\ref{LP:1}). Consequently, RTS is a randomized algorithm with expected approximation ratio 1.45.
\end{thm}

The proof, given in Appendix~\ref{supp:RTS}, verifies step by step that each ingredient of the analysis of Im and Li~\cite{Im2023} is unaffected by the contact delays: (1)~the delays $\Phi_{i}$ appear only in the objective of~(\ref{LP:1}), whose constraints coincide with those of~\cite{Im2023}; (2)~the grouping, the SNC rounding and the priorities do not depend on $\Phi_{i}$, so the joint distribution of the assignment and of the processing order is that of~\cite{Im2023}; (3)~conditioned on $\sigma(j)=i$, the completion time of task $j$ equals $\Phi_{i}$ plus its completion time in the schedule of~\cite{Im2023}; and (4)~the per-task bound of~\cite[inequality~(4.7)]{Im2023} then yields the ratio 1.45. We remark that the argument uses only the fact that the expected contact delay is an additive, order-independent term for each crowd worker; this is precisely what makes the extension possible. If the contact delay of a task depended on its position in the processing sequence (e.g., if all results were collected only after the last task is finished, or if the meeting opportunities were not stationary over time), step~(3) would fail and a different analysis would be required. Moreover, when $T$ is reduced to a polynomial size as discussed above, the guarantee becomes $1.45(1+\epsilon)$ for any fixed $\epsilon>0$.


\section{Deterministic Algorithm with Time-indexed Approach}\label{sec:Algorithm4}
RTS can be derandomized by the method of conditional expectations, which yields the \emph{Deterministic Time-indexed Scheduling} (DTS) algorithm. A direct derandomization computes, for each task, the expected total weighted completion time for every crowd worker and every time point, which takes $O(mn^2T^2)$ time. However, once the assignment of tasks to crowd workers is fixed, processing the tasks of each crowd worker in Smith-ratio order is optimal~\cite{smith1956}, so the processing order is determined by the assignment and the expected cost need not be evaluated at every time point. Nevertheless, since the groups of a crowd worker are not disjoint, the conditional expectations must still be computed for each group, which takes $O(mn^2G)$ time, where $G$ is the number of groups per crowd worker, and the random choices that define the groups ($\rho$ and the offsets $\tau_{ij}$) must also be derandomized by an exhaustive search over these parameters. The conditional expectations used by DTS and its pseudocode are given in Appendix~\ref{supp:DTS}. Because the grouping step is computationally expensive, Section~\ref{sec:Algorithm5} presents a deterministic algorithm that does not use groups.

\section{Efficient Deterministic Algorithm with Time-indexed Approach}\label{sec:Algorithm5}
This section presents the \emph{Efficient Deterministic Time-indexed Scheduling} (EDTS) algorithm (Algorithm~\ref{Alg5}). EDTS omits the time-consuming grouping step, so that it is the derandomization of independent rounding. Excluding the time needed to solve the linear program, the derandomized assignment and sequencing steps take only $O(mn^2)$ time; the overall running time is $O(mn^2)+LP(mnT)$, where $LP(mnT)$ denotes the time to solve the LP relaxation of~(\ref{LP:1}), which has $O(mnT)$ variables and $O(n+mT)$ constraints besides the non-negativity and boundary constraints. Removing the grouping mechanism also removes the strong negative correlation on which the 1.45 bound relies; EDTS therefore does not inherit the 1.45 guarantee. Instead, we prove in Theorem~\ref{thm:EDTS} that it is a deterministic 1.5-approximation algorithm. Let $\sigma: \mathcal{J}\rightarrow \mathcal{M}$ denote the assignment of tasks to crowd workers, and for every $i \in \mathcal{M}$ and $j \in \mathcal{J}$, let
\begin{eqnarray*}
y_{ij} = \sum_{s} x_{ijs}
\end{eqnarray*}
be the total height of the rectangles of task $j$ on crowd worker $i$. For a crowd worker $i$, we define a binary relation $\preceq_{i}$ on $\mathcal{J}$: for two tasks $j, j' \in \mathcal{J}$, $j' \preceq_{i} j$ if and only if$$\frac{w_{j'}}{p_{ij'}} > \frac{w_j}{p_{ij}} \text{ or } \left( \frac{w_{j'}}{p_{ij'}} = \frac{w_j}{p_{ij}} \text{ and } j' > j \right).$$ This relation follows Smith's rule: $j' \preceq_{i} j$ means that $j'$ is processed before $j$ on crowd worker $i$. Since the service times depend on the crowd worker, this order may differ between crowd workers. If each task $j$ is assigned to crowd worker $i$ independently with probability $y_{ij}$ and the tasks on each crowd worker are processed in Smith-ratio order, then
\begin{eqnarray*}
E[C_{j}|\sigma(j)=i] & = & \Phi_{i}+p_{ij} +\sum_{j'\in \mathcal{J}, j'\preceq_{i} j} y_{ij'}p_{ij'}.
\end{eqnarray*}
Accordingly, the expected completion time of task $j$ is given by:
\begin{eqnarray*}
E[C_{j}] & = & \sum_{i} y_{ij} E[C_{j}|\sigma(j)=i].
\end{eqnarray*}

Let $\mathcal{P} \subseteq \mathcal{J}$ be the set of tasks that have already been assigned. For each task $j$, the binary variable $z_{ij}$ indicates whether $j$ has been assigned to crowd worker $i$, and $\bar{z}_{j} = 1 - \sum_{i}z_{ij}$ indicates that $j$ is still unassigned. Given this partial assignment, the conditional expected completion time of task $j$ on crowd worker $i$ is

\begin{eqnarray}\label{eq:EDTS-cond}
E[C_{j}|\sigma(j)=i] & = & \Phi_{i}+p_{ij}+\sum_{j'\in \mathcal{J}, j'\preceq_{i} j} z_{ij'} p_{ij'} \nonumber \\
                     &   & +\sum_{j'\in \mathcal{J}, j'\preceq_{i} j} \bar{z}_{j'} y_{ij'}p_{ij'}.
\end{eqnarray}
For an unassigned task $j\notin\mathcal{P}$, the expected completion time is $E[C_{j}] = \sum_{i} y_{ij} E[C_{j}|\sigma(j)=i]$; for a task $j\in\mathcal{P}$ that has been assigned to crowd worker $i$, $E[C_{j}]=E[C_{j}|\sigma(j)=i]$ as given by~(\ref{eq:EDTS-cond}).

EDTS applies the method of conditional expectations to these quantities. Lines~\ref{alg5-2}--\ref{alg5-3} assign the tasks one by one to the crowd worker that minimizes the conditional expectation of the objective, and line~\ref{alg5-4} processes the tasks of each crowd worker in non-increasing order of the Smith ratio $w/p$.

\begin{algorithm}
\caption{The EDTS Algorithm}
\label{Alg5}
\begin{algorithmic}[1]
    \STATE Solve the linear program (\ref{LP:1}) to obtain $y$.
    \STATE Initialize $\mathcal{P} \leftarrow \emptyset$, $z \leftarrow 0$, and $\mathcal{J}_i \leftarrow \emptyset$ for all $i \in \mathcal{M}$. \label{alg5-2}

    \FOR{each task $j \in \mathcal{J}$}
        \STATE Select crowd worker $i^* \leftarrow \arg\min_{i} E[\sum_{q} w_q C_q \mid \mathcal{P} \cup \{j\}, z_{ij}=1]$.
        \STATE Update $\mathcal{P} \leftarrow \mathcal{P} \cup \{j\}$, $z_{i^*j} \leftarrow 1$, and $\mathcal{J}_{i^*} \leftarrow \mathcal{J}_{i^*} \cup \{j\}$.
    \ENDFOR \label{alg5-3}

    \FOR{each crowd worker $i \in \mathcal{M}$}
        \STATE Schedule tasks in $\mathcal{J}_i$ in non-increasing order of the Smith ratio ($w/p$). \label{alg5-4}
    \ENDFOR
\end{algorithmic}
\end{algorithm}

The following theorem establishes the approximation guarantee of the EDTS algorithm.
\begin{thm}\label{thm:EDTS}
Let $x$ be an optimal solution of the LP relaxation of~(\ref{LP:1}) and $LP^{*}$ its objective value. The EDTS algorithm outputs a schedule whose total weighted completion time is at most $1.5\cdot LP^{*}\leq 1.5\cdot WCT_{OPT}$.
\end{thm}

The proof is given in Appendix~\ref{supp:EDTS}. It compares EDTS with the independent rounding of Section~\ref{sec:Algorithm3}: both assign task $j$ to crowd worker $i$ independently with probability $y_{ij}$, for every realized assignment the Smith-ratio order on each crowd worker is no worse than the $\theta$-order of independent rounding, and the method of conditional expectations does not increase the expected objective.

Theorem~\ref{thm:EDTS} matches the guarantee of independent rounding while avoiding the grouping step of DTS. However, the time-indexed LP~(\ref{LP:1}) has $O(mnT)$ variables and is therefore only pseudo-polynomial in size. We next show that EDTS remains a constant-factor approximation algorithm with polynomial running time when LP~(\ref{LP:1}) is replaced by an interval-indexed LP; this is the version used in our experiments.

\subsection{Polynomial-Time Version with an Interval-Indexed LP}\label{sec:Algorithm5:interval}
We adopt the interval-indexed LP of Schulz and Skutella~\cite{schulz2002scheduling} and add the contact delays. We assume $p_{ij}>0$ for all $i\in\mathcal{M}$ and $j\in\mathcal{J}$, and let $p_{min}=\min_{i,j}p_{ij}$. For a given parameter $\epsilon>0$, let $\delta=\frac{\epsilon}{2}p_{min}$, $\tau_{-1}=0$ and $\tau_{\ell}=\delta(1+\epsilon)^{\ell}$ for $\ell\geq 0$. The time horizon is divided into the intervals $I_{\ell}=(\tau_{\ell-1},\tau_{\ell}]$ for $\ell\in\mathcal{L}=\{0,1,\ldots,L\}$, where
\begin{eqnarray*}
L= \left\lceil \log_{(1+\epsilon)} \left(\frac{1}{\delta}\sum_{j \in \mathcal{J}}\max_{i \in \mathcal{M}} \{p_{ij}\}\right)\right\rceil
\end{eqnarray*}
guarantees $\tau_{L}\geq\sum_{j}\max_{i}p_{ij}$. Hence $L$, and thus the size of the following LP, is polynomial in the input size for any fixed $\epsilon$. The variable $x_{ij\ell}$ is the fraction of task $j$ processed by crowd worker $i$ during interval $I_{\ell}$, so that $p_{ij}x_{ij\ell}$ is the corresponding amount of processing.
\begin{subequations}\label{LP:Interval}
\begin{align}
& \text{min}  && \sum_{j \in \mathcal{J}} w_{j} C_{j}    &   & \tag{\ref{LP:Interval}} \\
& \text{s.t.} && \sum_{i \in \mathcal{M}} \sum_{\ell \in \mathcal{L}} x_{ij\ell} = 1, && \forall j\in \mathcal{J} \label{LPI:a}  \\
& && \sum_{j \in \mathcal{J}} p_{ij} x_{ij\ell} \leq \tau_{\ell} - \tau_{\ell-1}, && \forall i, \ell \label{LPI:b}\\
& && C_{j}\geq \sum_{i \in \mathcal{M}}\sum_{\ell \in \mathcal{L}} x_{ij\ell}\Big(\Phi_{i}+\tau_{\ell-1}+\frac{p_{ij}}{2}\Big), && \forall j\in \mathcal{J} \label{LPI:c}\\
& && C_{j}\geq \sum_{i \in \mathcal{M}}\sum_{\ell \in \mathcal{L}} x_{ij\ell}\big(\Phi_{i}+p_{ij}\big), && \forall j\in \mathcal{J} \label{LPI:d}\\
& && x_{ij\ell} \geq 0, && \forall i, j, \ell \notag
\end{align}
\end{subequations}

\begin{lem}\label{lem:LPI}
The optimal value $LP_{I}^{*}$ of LP~(\ref{LP:Interval}) is a lower bound on the optimal total weighted completion time $WCT_{OPT}$.
\end{lem}

The proof, given in Appendix~\ref{supp:EDTS}, maps an optimal schedule to a feasible solution of LP~(\ref{LP:Interval}) and uses the fact that the mean busy time $S_{j}+p_{ij}/2$ of a task is at least $\sum_{\ell}x_{ij\ell}\tau_{\ell-1}$.

In the polynomial-time version of EDTS, we solve LP~(\ref{LP:Interval}) and set $y_{ij}=\sum_{\ell}x_{ij\ell}$ in Algorithm~\ref{Alg5}; all other steps are unchanged.

\begin{thm}\label{thm:EDTS-interval}
For any fixed $\epsilon>0$, the polynomial-time version of EDTS is a deterministic $(1.5+\epsilon/2)$-approximation algorithm.
\end{thm}

The proof (Appendix~\ref{supp:EDTS}) applies independent rounding to LP~(\ref{LP:Interval}), in which each task selects a pair $(i,\ell)$ with probability $x_{ij\ell}$ and a uniformly random time in $I_{\ell}$; since $\tau_{\ell}=(1+\epsilon)\tau_{\ell-1}$ for $\ell\geq 1$, the expected completion time of every task is at most $(1.5+\epsilon/2)$ times its LP value, and the derandomization argument of Theorem~\ref{thm:EDTS} applies unchanged.

LP~(\ref{LP:Interval}) has $O(mnL)$ variables and $O(n+mL)$ constraints, where $L=O(\frac{1}{\epsilon}\log\frac{\sum_{j}\max_{i}p_{ij}}{\epsilon\, p_{min}})$. Thus, the polynomial-time version of EDTS runs in $O(mn^2)+LP(mnL)$ time.

\section{Results and Discussion}\label{sec:Results}
This section compares the proposed EDTS algorithm with three LRF-based algorithms designed for the same model and with two general-purpose heuristics for unrelated machines. We use synthetic instances and instances whose contact delays are derived from real contact traces, and we compare all algorithms with optimal solutions on small instances. All algorithms were implemented in Python~3.13 with NumPy~2.3 and SciPy~1.16, and all linear and mixed-integer programs were solved with Gurobi~13.0 using its default parameters. All experiments were run on a desktop computer with an Intel Core Ultra~7 270K Plus processor (24 cores) and 32~GB of RAM running Windows~10.

The algorithms are evaluated using the Total Weighted Completion Time Ratio (WCTR), defined as the ratio of an algorithm's total weighted completion time to the optimal value of the interval-indexed LP~(\ref{LP:Interval}), in which all variables are continuous. Normalizing by an instance-specific value makes the results comparable across task sets with different characteristics. By Lemma~\ref{lem:LPI}, the LP value is a lower bound on (rather than equal to) the optimal total weighted completion time. The WCTR is therefore an upper estimate of the true ratio to the optimum, and its absolute values include the integrality gap of the LP. The WCTR is most informative for comparing algorithms on the same instances; Section~\ref{sec:Results:opt} quantifies the integrality gap on small instances that can be solved to optimality.

\subsection{Synthetic Datasets and Settings}
In our synthetic dataset, we evaluate the effects of several factors, including the task-to-worker ratio $\frac{n}{m}$, the number of crowd workers, the mean and the variance of the base service time of tasks, the range of crowd workers' base capabilities, and the range of task influence factors.
Below are the default values for these parameters:
\begin{itemize}
\item For each crowd worker $i$, the total contact time $\Phi_{i}$ is an integer drawn uniformly at random from $[1, 30]$; it is redrawn for every instance.
\item The total number of crowd workers is set to 10.
\item The ratio $\frac{n}{m}$ is established at 25.
\item To determine the required service times for tasks, we start by generating a base service time, denoted as $\alpha_{j}$, for each task $j$. We also establish a base capability, represented as $\beta_{i}$, for each crowd worker $i$. Additionally, we introduce an influence factor $\gamma_{ij}$, which represents crowd worker $i$'s proficiency in handling task $j$. The required service time is defined as $p_{ij} = \max\{1,\lfloor\alpha_{j} \beta_{i} \gamma_{ij}\rceil\}$~\cite{ali2000task}, where $\lfloor\cdot\rceil$ denotes rounding to the nearest integer. In our experimental setup, the base service time is $\alpha_{j}=|\lfloor X_{j}\rceil|+1$, where $X_{j}$ follows a Gaussian distribution whose mean and variance are both set to 30 time units by default. The base capability $\beta_{i}$ is sampled from a uniform distribution within the interval $[0.5, 2]$, while the influence factor $\gamma_{ij}$ is similarly sampled from a uniform distribution in the range $[0.1, 2]$.
\item The weight of each task is an integer drawn uniformly at random from $[1, 100]$.
\end{itemize}
In each experiment, one parameter is varied while the others are kept at their default values: $\frac{n}{m}\in\{5,10,\ldots,50\}$; $m\in\{5,10,\ldots,25\}$ with $\frac{n}{m}=25$; the mean of $X_{j}$ in $\{5,10,\ldots,50\}$; the variance of $X_{j}$ in $\{20,22,\ldots,38\}$; $\beta_{i}$ uniform in $[0.2k,2]$; and $\gamma_{ij}$ uniform in $[0.2,2-0.2(k-1)]$ for $k=1,\ldots,10$, so that the ratio of the lower and upper limits of $\beta$ (respectively, $\gamma$) ranges from 0.1 to 1. For each setting we generate 100 instances and report the mean and the standard deviation of the WCTR. Every instance is generated from its own random seed, so that each instance can be reproduced independently.

\subsection{Real Datasets and Settings}
We follow the methodology established in previous studies~\cite{Zhang2025, chen2025approx} and conduct our experiments using the Cambridge Haggle dataset~\cite{c70011-22}. This dataset contains records of Bluetooth contacts between mobile devices carried by participants over several days, in three scenarios: an office environment (Intel), an academic environment (Cambridge), and a conference (Infocom). Each record contains the identifiers of the two devices and the times of the first and last sightings of the contact, which allows us to estimate the parameters of the exponential distribution. In our experimental setup, each iMote device acts in turn as the requester $r$, and the other devices it encountered are treated as crowd workers. For each crowd worker $i$, we define $\lambda_{i}$ as follows:
\begin{eqnarray*}
\lambda_{i} = \frac{l_{i}}{\sum_{k=1}^{l_{i}}\Delta t_{k,i}},
\end{eqnarray*}
where $l_{i}$ is the total number of contacts between requester $r$ and crowd worker $i$, $\Delta t_{k,i}$ is the inter-meeting time between the end of the $(k-1)$-th contact and the start of the $k$-th contact, and $\Delta t_{1,i}$ is measured from the start of the trace. We set $\phi_{i}=1/\lambda_{i}$ and use one hour as the time unit, i.e., $\Phi_{i}=\lceil 2/\lambda_{i}\rceil$ with $\lambda_{i}$ measured in contacts per hour, so that the contact delays and the service times $p_{ij}$ are of comparable magnitude. For each requester, we select the 128 devices with the highest $\lambda_{i}$ values as its crowd workers (all encountered devices if there are fewer than 128), and we vary $n/m$ from 2 to 6. For every value of $n/m$, about 100 instances are generated and divided evenly among the requesters of the dataset, and the results are averaged over all requesters and instances.
Note that only the contact process is taken from the traces: the service times $p_{ij}$ and the task weights $w_{j}$ are generated as in the synthetic setting. The real-data experiments therefore evaluate the algorithms under realistic contact heterogeneity, not under real crowdsourcing workloads.

\subsection{Algorithms}\label{sec:Results:algorithms}
Because the time-indexed LP~(\ref{LP:1}) is only pseudo-polynomial in size, we use the polynomial-time version of EDTS described in Section~\ref{sec:Algorithm5:interval}, which solves the interval-indexed LP~(\ref{LP:Interval}) and sets $y_{ij}=\sum_{\ell}x_{ij\ell}$. By Theorem~\ref{thm:EDTS-interval}, this implementation is a $(1.5+\epsilon/2)$-approximation algorithm, and by Lemma~\ref{lem:LPI} the optimal value of LP~(\ref{LP:Interval}) is a valid lower bound used to normalize the WCTR. We set $\epsilon = 0.5$ in all experiments. With this value, LP~(\ref{LP:Interval}) has 27 to 32 intervals on our synthetic instances; for the default setting ($m=10$, $n=250$) it has about 70,000 variables and 1,000 constraints, and for the largest synthetic setting ($m=25$, $n=625$) about 485,000 variables and 2,700 constraints.

The first group of baselines consists of the LRF-based algorithms of~\cite{Zhang2025, chen2025approx}, which are the existing algorithms designed for the same mobile crowdsourcing model with contact delays. Because the original LRF algorithm was designed for identical crowd workers, we need to standardize the required service times for each task. Let $p_j$ represent the required service time of task $j$, which serves exclusively to determine scheduling priority. We evaluate three variants, LRF-MAX, LRF-MIN, and LRF-MEAN, defined as follows:
\begin{itemize}
\item \textbf{LRF-MAX:} $p_j=\max_{i \in \mathcal{M}} \{p_{ij}\}$.
\item \textbf{LRF-MIN:} $p_j=\min_{i \in \mathcal{M}} \{p_{ij}\}$.
\item \textbf{LRF-MEAN:} $p_j=\text{mean}_{i \in \mathcal{M}} \{p_{ij}\}$.
\end{itemize}
Each algorithm determines the scheduling sequence by sorting tasks in non-increasing order of their Smith ratios. Finally, each task is assigned to the crowd worker that minimizes its completion time.

The second group consists of two classical general-purpose heuristics for scheduling unrelated machines to minimize the total weighted completion time~\cite{vredeveld2002experimental}. Since $\Phi_{i}$ is a constant added to the completion time of every task assigned to crowd worker $i$, the tasks on each crowd worker are always sequenced by Smith's rule.
\begin{itemize}
\item \textbf{GREEDY:} The tasks are considered in non-increasing order of $\max_{i}w_{j}/p_{ij}$. Each task is inserted into the crowd worker for which the total weighted completion time increases the least, where the increase accounts both for the completion time of the inserted task and for the delay it imposes on the tasks sequenced after it.
\item \textbf{GREEDY-LS:} Starting from the GREEDY schedule, a local search repeatedly moves a single task to another crowd worker whenever this strictly decreases the total weighted completion time, until no improving move exists.
\end{itemize}
Unlike EDTS, neither heuristic has a worst-case performance guarantee.

Finally, to evaluate the algorithms against the true optimum, we solve small instances exactly. Since the Smith-ratio order is optimal on each crowd worker once the assignment is fixed, an optimal schedule is obtained from a mixed-integer program with assignment variables $x_{ij}\in\{0,1\}$ and variables $u_{ikj}\geq x_{ik}+x_{ij}-1$, $u_{ikj}\geq 0$, for every pair of tasks $k,j$ such that $k$ precedes $j$ in the Smith order on crowd worker $i$; its objective is $\sum_{i,j}w_{j}(\Phi_{i}+p_{ij})x_{ij}+\sum_{i}\sum_{k\prec_{i}j}w_{j}p_{ik}u_{ikj}$.

\subsection{Results on Synthetic Datasets}
\begin{figure*}[!t]
\centering
\subfloat[]{\includegraphics[width=0.325\textwidth]{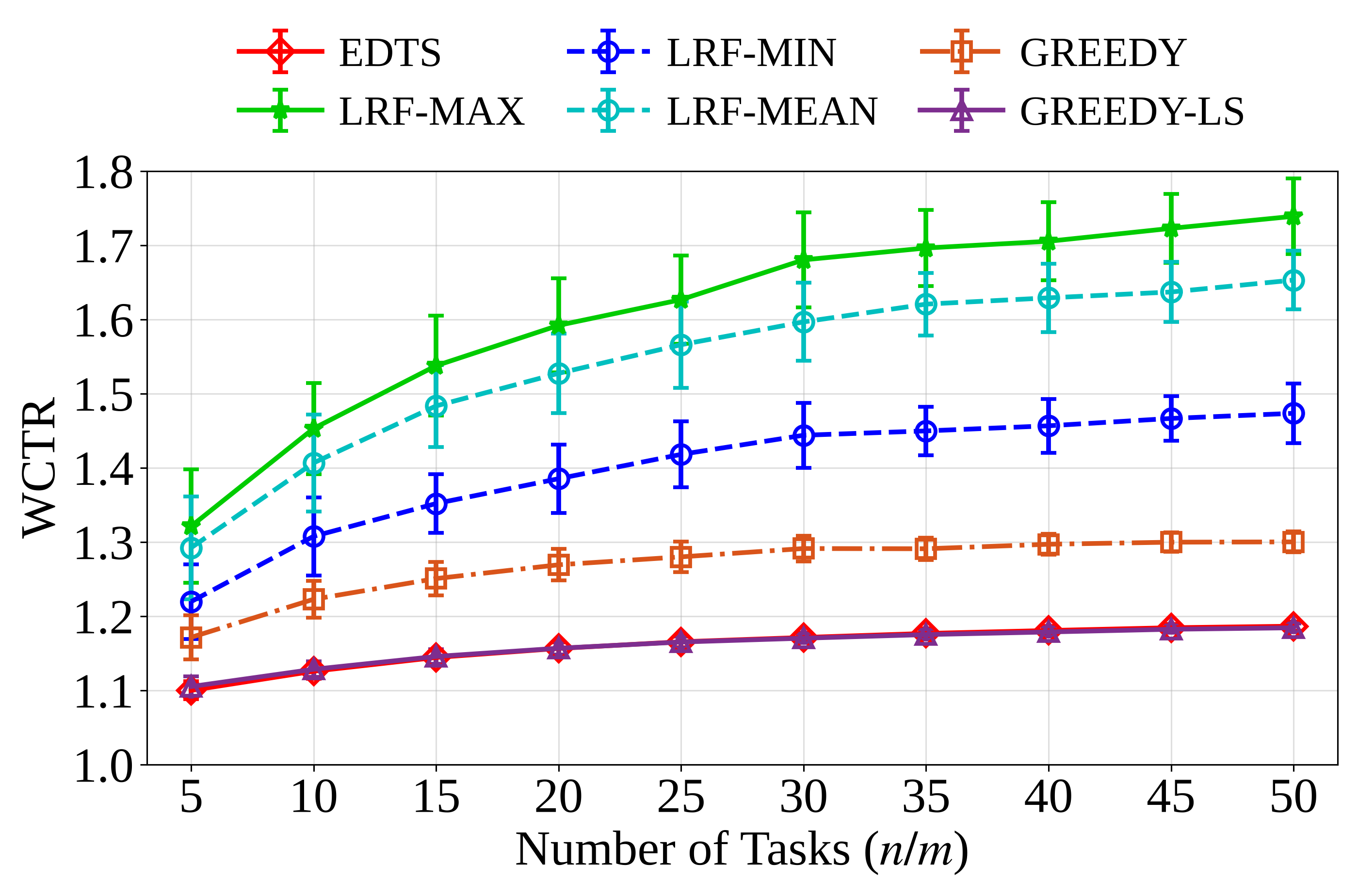}\label{fig:ratio2}}
\hfil
\subfloat[]{\includegraphics[width=0.325\textwidth]{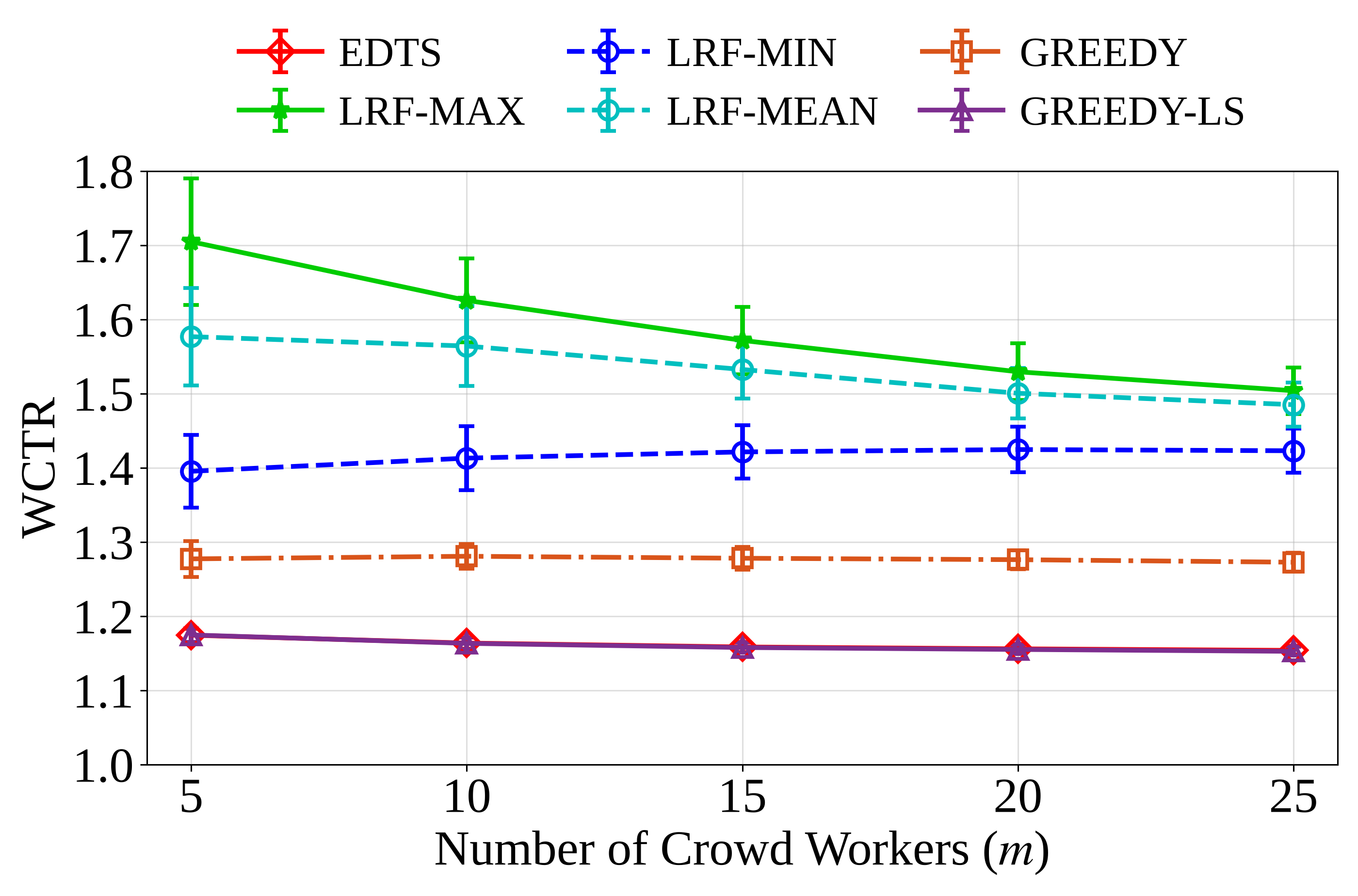}\label{fig:ratio3}}
\hfil
\subfloat[]{\includegraphics[width=0.325\textwidth]{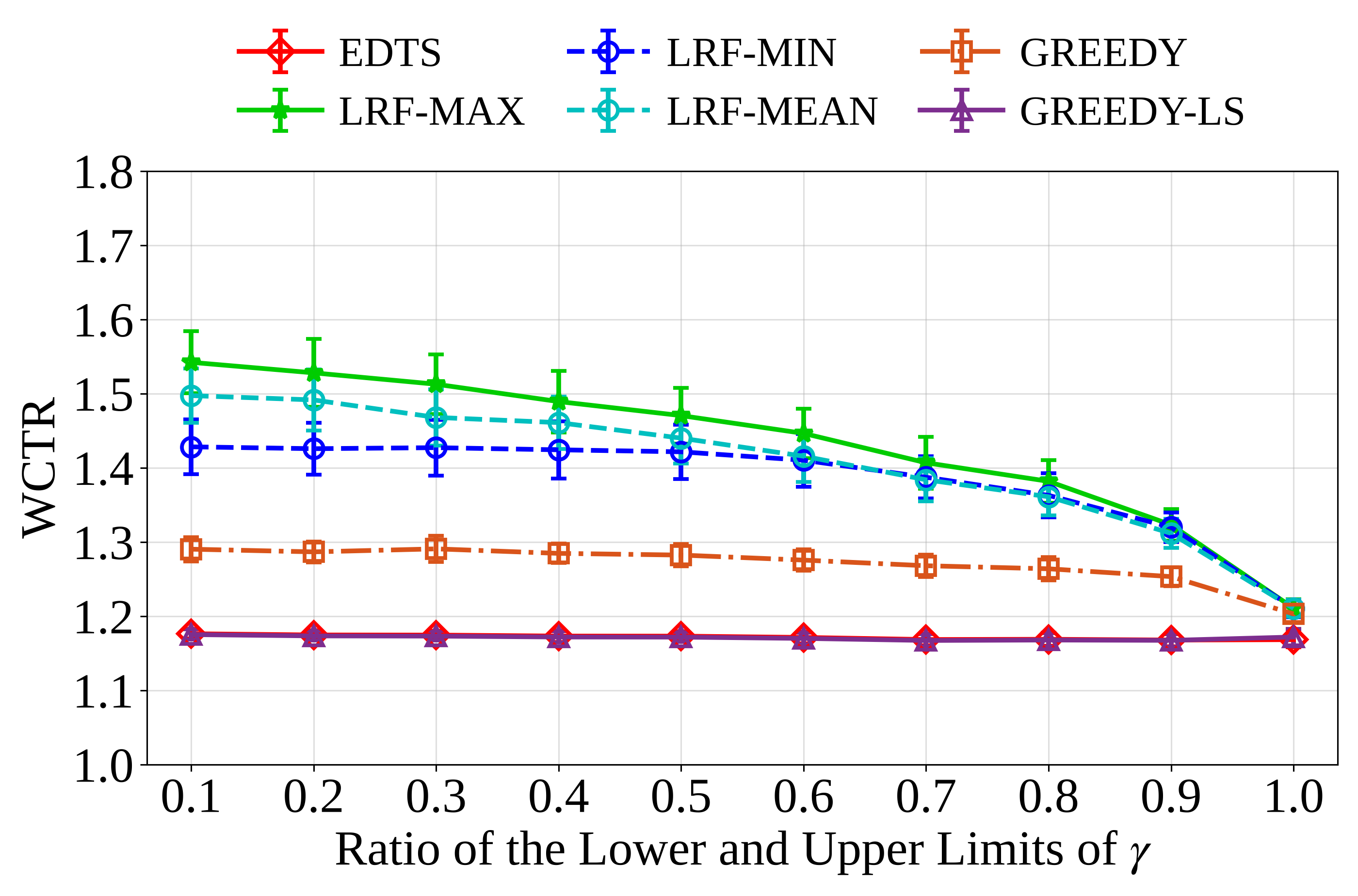}\label{fig:ratio1_2}}
\caption{WCTR of the six algorithms on synthetic instances versus (a) the number of tasks per crowd worker ($n/m$), (b) the number of crowd workers ($m$), and (c) the ratio of the lower and upper limits of $\gamma$.}
\label{fig:syn}
\end{figure*}

Figure~\ref{fig:ratio2} illustrates the WCTR as the number of tasks per crowd worker ($n/m$) increases from 5 to 50. EDTS achieves a WCTR between 1.100 and 1.187 with standard deviations of at most 0.012. The WCTR of all LRF variants grows considerably faster with the task density: the best variant, LRF-MIN, increases from 1.220 to 1.474, and LRF-MAX reaches 1.739 at $n/m=50$. GREEDY lies between these two groups (1.172 to 1.300). GREEDY-LS performs almost identically to EDTS: EDTS is slightly better for small task densities ($n/m\leq 15$, where it yields the smaller value in 70\% to 73\% of the instances), while GREEDY-LS is slightly better for large task densities, and the difference between their mean WCTRs never exceeds 0.5\%. At $n/m=50$, EDTS and GREEDY-LS also have the smallest interquartile ranges of the WCTR, four to seven times smaller than those of the LRF variants (Fig.~\ref{fig:ratio4_box} in the appendix).

Figure~\ref{fig:ratio3} evaluates the impact of the number of crowd workers ($m$) with $n/m=25$. The WCTR of EDTS decreases slightly from 1.175 to 1.154 as $m$ increases from 5 to 25. LRF-MAX and LRF-MEAN improve with more crowd workers (LRF-MAX from 1.705 to 1.504), whereas LRF-MIN remains between 1.395 and 1.425. EDTS reduces the WCTR by 16\% to 19\% compared with the best LRF variant, and the mean WCTRs of EDTS and GREEDY-LS differ by at most 0.11\%.

Figure~\ref{fig:ratio1_2} evaluates the impact of the ratio of the lower and upper limits of $\gamma$. As the ratio approaches 1, $\gamma_{ij}$ becomes nearly constant, so that $p_{ij}\approx\alpha_{j}\beta_{i}\gamma$ and the instance is close to a related-worker instance, for which the list-scheduling rule of LRF is well suited. Accordingly, the WCTR of all LRF variants decreases sharply, from 1.428 to 1.211 for the best variant, and that of GREEDY decreases from about 1.29 to 1.203. EDTS is much less sensitive to this parameter (1.168 to 1.177) and remains better than GREEDY and all LRF variants over the whole range, with GREEDY-LS within 0.3\% of it, but its advantage over the best LRF variant shrinks from about 18\% for ratios up to 0.5 to 3.5\% when the ratio equals 1. Hence, the benefit of EDTS over LRF is largest when the crowd workers are substantially heterogeneous with respect to individual tasks.

Appendix~\ref{supp:exp} reports the effects of the mean and the variance of the base service time $\alpha$ and of the heterogeneity of the base capabilities $\beta$ (Figs.~\ref{fig:ratio4}--\ref{fig:ratio1}). In these experiments, EDTS and GREEDY-LS again perform almost identically, both clearly outperform GREEDY, and EDTS reduces the WCTR by 15\% to 20\% compared with the best LRF variant in most settings; the variance of $\alpha$ has little influence on the relative performance of the algorithms.

\subsection{Results on Real Datasets}
\begin{figure*}[!t]
\centering
\subfloat[]{\includegraphics[width=0.325\textwidth]{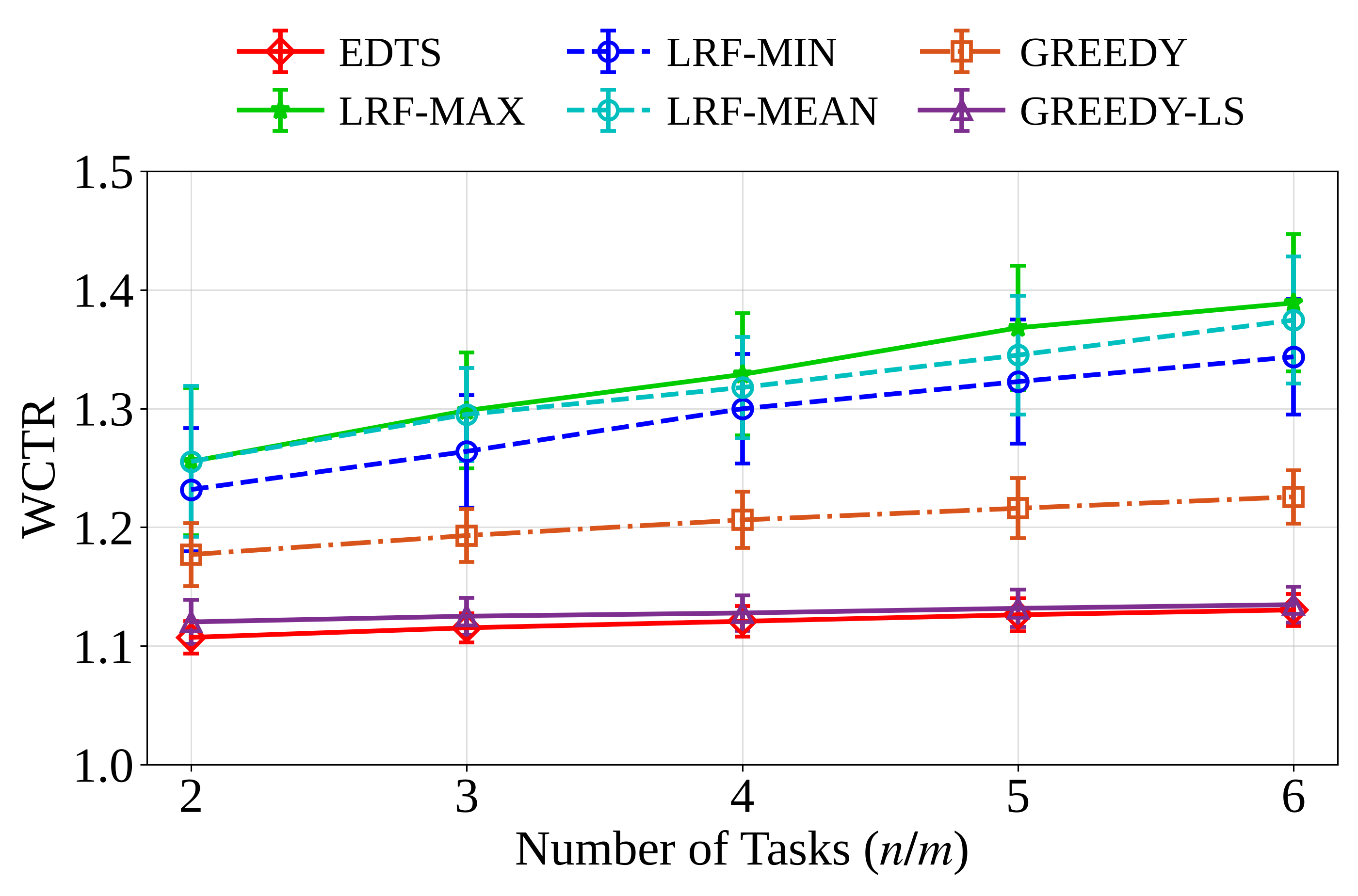}\label{fig:ratio6}}
\hfil
\subfloat[]{\includegraphics[width=0.325\textwidth]{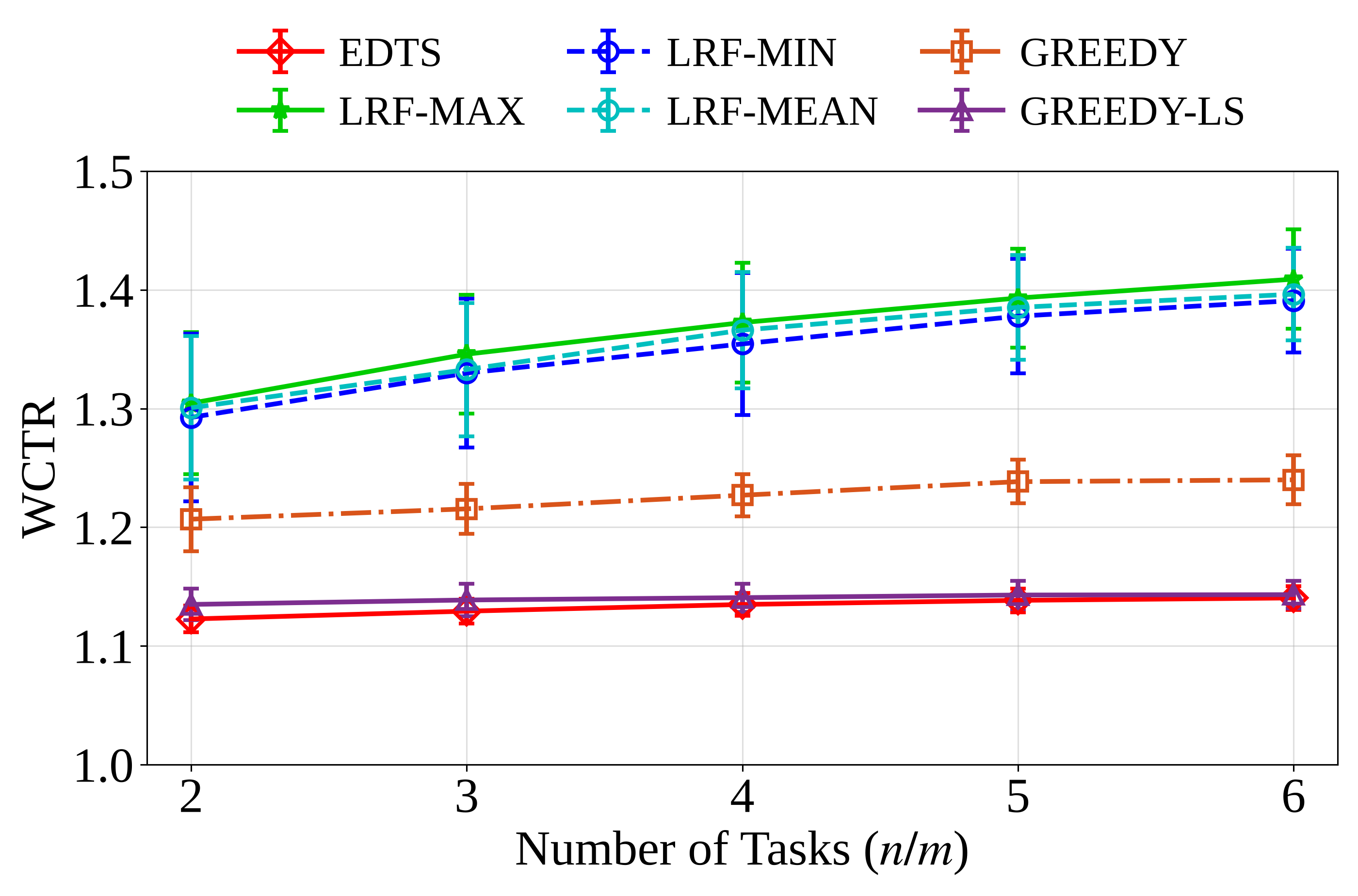}\label{fig:ratio7}}
\hfil
\subfloat[]{\includegraphics[width=0.325\textwidth]{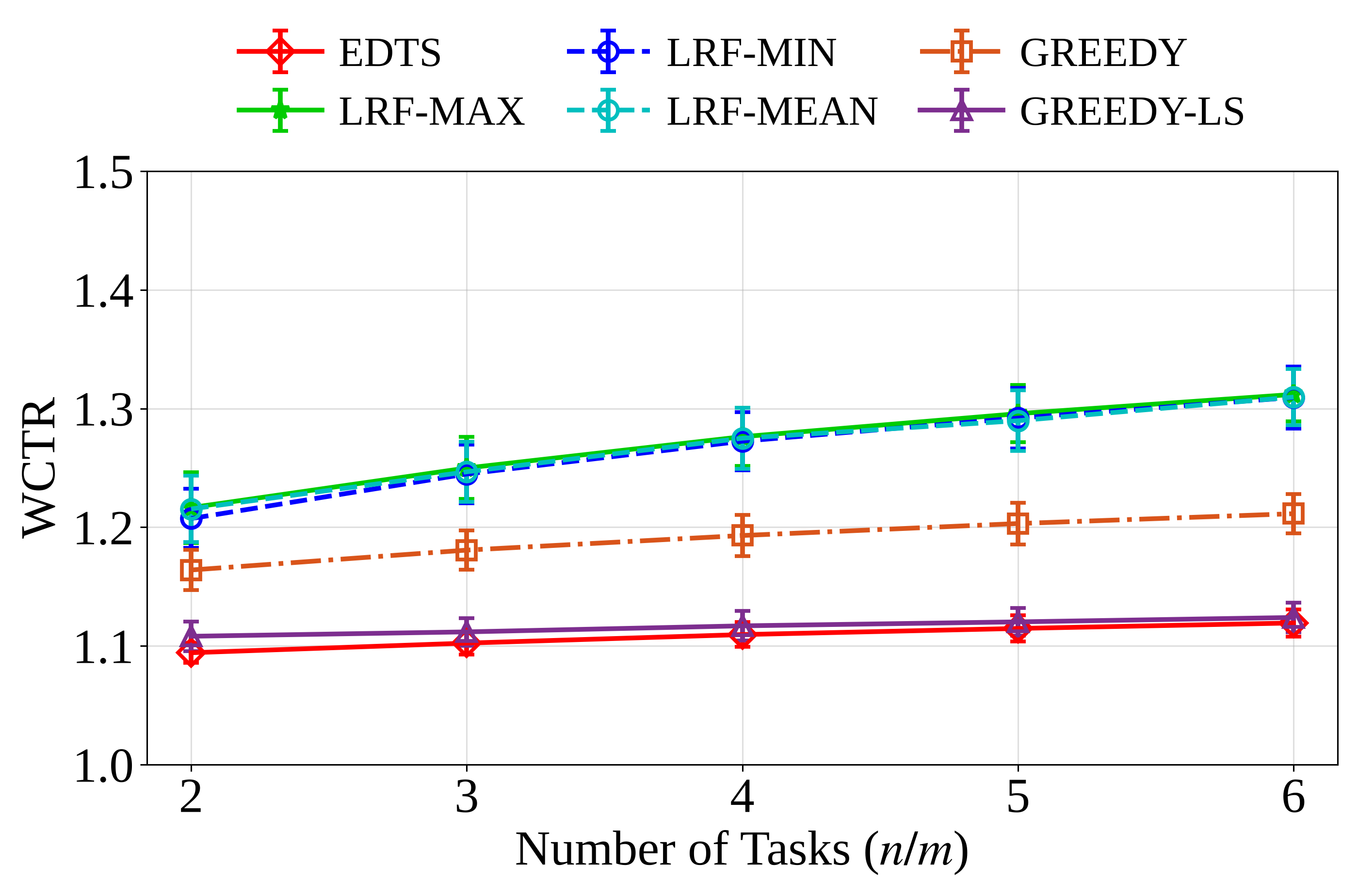}\label{fig:ratio8}}
\caption{WCTR versus $n/m$ on the real contact traces: (a) Intel, (b) Cambridge, and (c) Infocom.}
\label{fig:real}
\end{figure*}

Figures~\ref{fig:ratio6}--\ref{fig:ratio8} present the results on the three real-world datasets, where $n/m$ varies from 2 to 6 and the number of crowd workers per requester ranges from 19 to 101. On all three datasets, EDTS achieves the lowest WCTR for every value of $n/m$. On the Intel dataset (Fig.~\ref{fig:ratio6}), the WCTR of EDTS ranges from 1.107 to 1.130, compared with 1.232 to 1.344 for the best LRF variant, a reduction of 10\% to 16\%. On the Cambridge dataset (Fig.~\ref{fig:ratio7}), the WCTR of EDTS ranges from 1.123 to 1.141 (best LRF variant: 1.293 to 1.391, a reduction of 13\% to 18\%), and on the Infocom dataset (Fig.~\ref{fig:ratio8}) it ranges from 1.094 to 1.119 (best LRF variant: 1.207 to 1.309, a reduction of 9\% to 15\%). In contrast to the synthetic instances, EDTS also outperforms GREEDY-LS on the real traces: its mean WCTR is 0.2\% to 1.3\% lower, and it yields the smaller value in 70\% to 100\% of the instances (at least 85\% of the instances on Intel and Infocom). The standard deviations of EDTS (at most 0.014) are also much smaller than those of the LRF variants (up to 0.071). In all three datasets, the gap between EDTS and the LRF variants widens as the task density increases, whereas the gap between EDTS and GREEDY-LS narrows.

\subsection{Comparison with Optimal Solutions}\label{sec:Results:opt}
To assess how far the algorithms are from the optimum, we solve small instances with $m=5$ and $n/m\in\{1,\ldots,6\}$ (30 instances each, other parameters at their default values) to proven optimality with the mixed-integer program of Section~\ref{sec:Results:algorithms}. Figure~\ref{fig:ratio_opt} shows the ratio of each algorithm's total weighted completion time to the optimum, together with the ratio of the LP value to the optimum. For each value of $n/m$, we average the ratio over the 30 instances. The average ratio of EDTS lies between 1.0025 and 1.0040 for every $n/m$, i.e., EDTS is on average within 0.40\% of the optimum, and the largest ratio of EDTS over all 180 instances is 1.044. The corresponding average ratios lie between 1.009 and 1.015 for GREEDY-LS, between 1.023 and 1.064 for GREEDY, and between 1.032 and 1.094 for the best LRF variant (LRF-MIN), whereas the average ratio of LRF-MAX reaches 1.228. Moreover, the LP value of~(\ref{LP:Interval}) is on average 4\% to 11\% below the optimum (LP/OPT decreases from 0.961 to 0.894 as $n/m$ increases). This shows that the WCTR values reported above, which are normalized by the LP value, mainly reflect the integrality gap of the LP; the actual distance of EDTS from the optimum is much smaller.

\begin{figure}[!ht]
    \centering
        \includegraphics[width=3.8in]{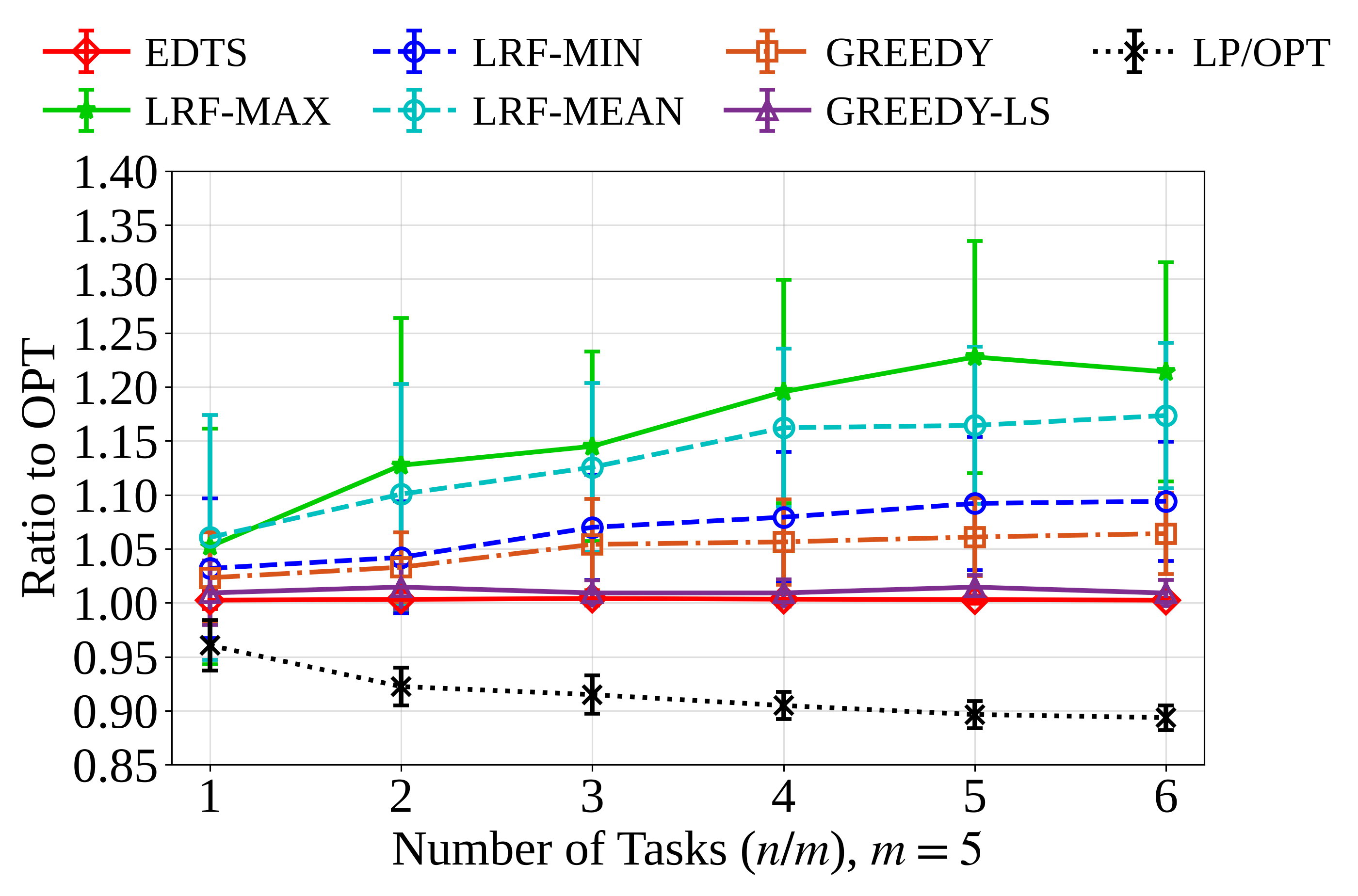}
    \caption{Ratio to the optimal total weighted completion time on small instances ($m=5$), and ratio of the LP value to the optimum.}
    \label{fig:ratio_opt}
\end{figure}

\subsection{Running Time}\label{sec:Results:time}
Table~\ref{tab:time} reports the average running times on representative instances. EDTS takes 0.65~s on default instances and at most 12.1~s on the largest instances we tested. In our implementation, the derandomized assignment step, written in pure Python, accounts for 54\% to 74\% of the running time of EDTS, whereas Gurobi solves LP~(\ref{LP:Interval}) in at most 4.6~s; the assignment step could be accelerated considerably by a compiled implementation. The LRF variants take at most 2~ms and GREEDY at most 0.15~s, and GREEDY-LS takes 0.14 to 1.8~s, so EDTS is about five to seven times slower than GREEDY-LS. The mixed-integer program of Section~\ref{sec:Results:algorithms} solves the small instances with $m=5$ in 0.08~s for $n=15$ and 1.58~s for $n=30$ on average, but its running time grows rapidly with $n$, which is why optimal solutions are reported only for small instances.

\begin{table}[!t]
\caption{Average Running Time in Seconds (5 Instances per Row)}
\label{tab:time}
\centering
\footnotesize
\setlength{\tabcolsep}{3pt}
\begin{tabular}{|l|r|r|r|r|r|r|}
\hline
Instance & $n$ & \multicolumn{2}{c|}{EDTS} & LRF & GREEDY & GREEDY-LS \\
\cline{3-4}
 & & LP & total & & & \\
\hline
$m=10$, $n/m=25$ & 250 & 0.30 & 0.65 & 0.000 & 0.013 & 0.138 \\
$m=10$, $n/m=50$ & 500 & 0.64 & 2.41 & 0.001 & 0.045 & 0.431 \\
$m=25$, $n/m=25$ & 625 & 2.52 & 8.68 & 0.002 & 0.115 & 1.453 \\
Infocom, $n/m=6$ & 456 & 4.61 & 12.10 & 0.001 & 0.145 & 1.824 \\
\hline
\end{tabular}
\end{table}

\subsection{Discussion}
In summary, EDTS outperforms all three LRF variants in every setting tested, reducing the WCTR by 3.5\% to 20\% compared with the best variant, and it is close to optimal on the small instances that can be solved exactly. The general-purpose heuristic GREEDY-LS achieves almost the same quality on the synthetic instances, and is slightly better than EDTS in some settings with many tasks per crowd worker, whereas EDTS is better on all three real traces. The main advantage of EDTS over GREEDY-LS is therefore not a consistently lower objective value but its provable worst-case guarantee of $1.5+\epsilon/2$ (Theorem~\ref{thm:EDTS-interval}), which GREEDY-LS lacks, together with comparable empirical performance, at the price of a running time that is several times longer (Section~\ref{sec:Results:time}).

\section{Concluding Remarks}\label{sec:Conclusion}
This paper studied task scheduling for minimizing the total weighted completion time in mobile crowdsourcing over mobile social networks. For identical crowd workers, we tightened the analysis of the LRF algorithm, removing the dependence on the task weights and obtaining the approximation ratio $\max\{\frac{3}{2}, \frac{\phi_{max}}{\phi_{min}}\}$, and we updated the competitive ratio of its online version. For unrelated crowd workers, we showed that the rounding scheme of Im and Li~\cite{Im2023} remains valid with worker-dependent contact delays, which gives an expected approximation ratio of 1.45 and improves upon the expected ratio of $1.5+\epsilon$ of our previous work for identical crowd workers, and we derandomized this algorithm. Finally, by exploiting the optimality of the Smith-ratio order on each crowd worker, we developed an efficient deterministic 1.5-approximation algorithm whose polynomial-time version achieves $1.5+\epsilon/2$. In our experiments, this algorithm outperforms the existing LRF variants in all tested settings, with the largest advantage when the crowd workers are highly heterogeneous, performs on par with a local-search heuristic for unrelated machines, and is close to optimal on small instances.

This work has several limitations that suggest directions for future research. First, the offline model assumes that all tasks are available at time zero, that the set of crowd workers is fixed, and that contacts are characterized by stationary expected inter-contact times; extending the analysis to dynamic task arrivals, worker churn, and stochastic or time-varying contact processes would broaden its applicability. Second, optimal solutions could be computed only for small instances, so on larger instances the algorithms are evaluated relative to an LP lower bound, whose integrality gap is about 10\% on the small instances. Third, the real-data experiments use real contact traces but synthetic task workloads, and evaluations with real crowdsourcing workloads remain to be carried out.


%


\section*{Acknowledgments}
This research was supported in part by the National Science and Technology Council, R.O.C. under Grant no. NSTC 115-2221-E-006-196.

The author used Claude (Anthropic), a generative AI assistant, in the preparation of this manuscript, as follows. (i) Text: the AI system drafted revisions of the abstract, the contributions and Table~\ref{tab:comparison} in Section~\ref{sec:introduction}, the description of the communication model in Section~\ref{sec:Preliminaries}, the proof of Theorem~\ref{thm:thm44} (Appendix~\ref{supp:RTS}), Theorem~\ref{thm:EDTS} and Section~\ref{sec:Algorithm5:interval} (LP~(\ref{LP:Interval}), Lemma~\ref{lem:LPI}, Theorem~\ref{thm:EDTS-interval} and their proofs), the text of Section~\ref{sec:Results} and of the additional results in the appendices, and the concluding remarks, and it was used for language editing throughout. (ii) Figures: the AI system generated the TikZ code of Fig.~\ref{fig:system} and the plotting code for the figures of the experimental results. (iii) Code: the AI system wrote or modified the experiment code (the interval-indexed LP, the GREEDY and GREEDY-LS heuristics, the mixed-integer program for optimal solutions, the estimation of the contact delays from the traces, and the experiment scripts); all experiments were run by the author. The author reviewed, verified and revised all AI-generated content, including every proof, and takes full responsibility for the content of this article.

\clearpage
\appendix

\section{Proofs for Section~\ref{sec:LRF}}\label{supp:LRF}
\subsection{Pseudocode of LRF and Proof of Theorem~\ref{thm:thm2}}
Algorithm~\ref{Alg_LRF} gives the pseudocode of the LRF algorithm described in Section~\ref{sec:LRF}: lines~\ref{Alg_LRF:1}--\ref{Alg_LRF:2} initialize the variables, and lines~\ref{Alg_LRF:3}--\ref{Alg_LRF:4} assign the tasks in Smith-ratio order to the crowd worker with the minimum expected workload.
\begin{algorithm}[htbp]
\caption{The LRF Algorithm}
\label{Alg_LRF}
\begin{algorithmic}[1]
    \REQUIRE Set of crowd workers $\mathcal{M}$, set of tasks $\mathcal{J}$
    \ENSURE Final partition $\Lambda_{LRF} = \{\mathcal{J}_1, \mathcal{J}_2, \dots, \mathcal{J}_m\}$

    \STATE \textbf{Initialize:}
    \FOR{each crowd worker $i \in \mathcal{M}$} \label{Alg_LRF:1}
        \STATE $\mathcal{J}_i \leftarrow \emptyset$
        \STATE $EW_i \leftarrow 2\phi_{i}$
    \ENDFOR \label{Alg_LRF:2}

    \STATE \textbf{Main Loop:}
    \FOR{each task $j \in \mathcal{J}$ in non-increasing order of the Smith ratio} \label{Alg_LRF:3}
        \STATE $i^* \leftarrow \arg\min_{i \in \mathcal{M}} \{ EW_i \}$ \label{Alg_LRF:select}
        \STATE $\mathcal{J}_{i^*} \leftarrow \mathcal{J}_{i^*} \cup \{j\}$
        \STATE $EW_{i^*} \leftarrow EW_{i^*} + p_j$
    \ENDFOR \label{Alg_LRF:4}

    \RETURN $\Lambda_{LRF} = \{\mathcal{J}_1, \mathcal{J}_2, \dots, \mathcal{J}_m\}$
\end{algorithmic}
\end{algorithm}

Throughout this subsection the crowd workers are identical, i.e., $p_{ij}=p_{j}$ for all $i$ and $j$, and the tasks are indexed in non-increasing order of their Smith ratios.
We use the following elementary fact.
\begin{fact}\label{fact:1}
Given $a, a', b, b'>0$, we have $\frac{a+a'}{b+b'}\leq \max\left\{\frac{a}{b},\frac{a'}{b'}\right\}$.
\end{fact}

The following two lemmas give a lower bound on the optimal cost and an upper bound on the cost of LRF, respectively.
\begin{lem}\label{lem:lem1}
Let $WCT_{OPT}$ be the minimum total weighted completion time. Then
\begin{eqnarray*}
WCT_{OPT} \geq \frac{1}{m}M_{1}+\frac{m-1}{2m}M_{n}+M_{\Lambda^*}
\end{eqnarray*}
where
\begin{eqnarray*}
M_{1}=\sum_{j=1}^{n}\sum_{k\leq j} w_{j}p_{k}, M_{n}=\sum_{j=1}^{n}w_{j}p_{j},
\end{eqnarray*}
\begin{eqnarray*}
M_{\Lambda^*}=  2\cdot\phi_{min} \sum_{j=1}^{n}w_{j}.
\end{eqnarray*}
\end{lem}
\begin{proof}
Let $OPT$ be the optimal value of the problem without communication delays. By Eastman \textit{et al.}~\cite{eastman1964},
\begin{eqnarray*}
OPT \geq \frac{1}{m}M_{1}+\frac{m-1}{2m}M_{n}
\end{eqnarray*}
With communication delays, every task $j$ incurs an additional cost of at least $2\phi_{min}w_{j}$. Therefore, $WCT_{OPT} \geq OPT + M_{\Lambda^*}$.
\end{proof}

\begin{lem}\label{lem:lem2}
Let $WCT_{LRF}$ be the total weighted completion time of the schedule produced by LRF. Then
\begin{eqnarray*}
WCT_{LRF} <  \frac{1}{m}M_{1}+\frac{m-1}{m}M_{n}+M_{\Lambda}
\end{eqnarray*}
where
\begin{eqnarray*}
M_{1}=\sum_{j=1}^{n}\sum_{k\leq j} w_{j}p_{k}, M_{n}=\sum_{j=1}^{n}w_{j}p_{j},
\end{eqnarray*}
\begin{eqnarray*}
M_{\Lambda}= 2\cdot\phi_{max} \sum_{j=1}^{n}w_{j}.
\end{eqnarray*}
\end{lem}
\begin{proof}
Let $a_{l}$ denote the current expected workload of crowd worker $l$. Task $j$ is assigned to a crowd worker $l$ with
\begin{eqnarray*}
a_{l} & =   & \min_{1\leq k \leq m} a_{k} \\
      &\leq & \frac{1}{m} \sum_{1\leq k \leq m} a_{k} \\
      &\leq & 2\cdot\phi_{max}+\frac{1}{m} \sum_{k\leq j} p_{k}. \\
\end{eqnarray*}

After all $n$ tasks have been assigned,
\begin{eqnarray*}
WCT_{LRF} < \frac{1}{m}(M_{1}-M_{n})+M_{n}+M_{\Lambda}
\end{eqnarray*}
because the contribution of task $j$ to the objective is at most
\begin{eqnarray*}
w_{j}\left(2\cdot\phi_{max}+\frac{1}{m} \sum_{k\leq j} p_{k}+p_{j}\right)
\end{eqnarray*}
for $m<j\leq n$ and
\begin{eqnarray*}
w_{j}\left(2\cdot\phi_{max}+p_{j}\right)
\end{eqnarray*}
for $1\leq j\leq m$.
\end{proof}

\begin{restated}{Theorem}{thm:thm2}For identical crowd workers, the LRF algorithm satisfies
\begin{eqnarray*}
\frac{WCT_{LRF}}{WCT_{OPT}} \leq \max\left\{\frac{3}{2},\frac{\phi_{max}}{\phi_{min}}\right\}.
\end{eqnarray*}
\end{restated}
\begin{proof}[Proof of Theorem~\ref{thm:thm2}]
By Lemmas~\ref{lem:lem1} and~\ref{lem:lem2},
\begin{eqnarray*}\label{ineq:03}
WCT_{LRF} < \frac{1}{m}M_{1}+\frac{m-1}{m}M_{n}+M_{\Lambda}
\end{eqnarray*}
and
\begin{eqnarray*}\label{ineq:04}
WCT_{OPT} \geq \frac{1}{m}M_{1}+\frac{m-1}{2m}M_{n}+M_{\Lambda^*}
\end{eqnarray*}
where
\begin{eqnarray*}\label{ineq:05}
M_{1}=\sum_{j=1}^{n}\sum_{k\leq j} w_{j}p_{k}, M_{n}=\sum_{j=1}^{n}w_{j}p_{j},
\end{eqnarray*}
\begin{eqnarray*}\label{ineq:06}
M_{\Lambda}= 2\cdot\phi_{max}\sum_{j=1}^{n}w_{j}, M_{\Lambda^*}= 2\cdot\phi_{min}\sum_{j=1}^{n}w_{j}.
\end{eqnarray*}
Then we have
\begin{eqnarray*}\label{ineq:07}
\frac{WCT_{LRF}}{WCT_{OPT}} & < &\frac{\frac{1}{m}M_{1}+\frac{m-1}{m}M_{n}+M_{\Lambda}}{\frac{1}{m}M_{1}+\frac{m-1}{2m}M_{n}+M_{\Lambda^*}} \\
 &\leq & \max\left\{\frac{\frac{1}{m}M_{1}+\frac{m-1}{m}M_{n}}{\frac{1}{m}M_{1}+\frac{m-1}{2m}M_{n}},\frac{M_{\Lambda}}{M_{\Lambda^*}}\right\}.
\end{eqnarray*}
The second inequality follows from Fact~\ref{fact:1}. Then we have
\begin{eqnarray*}
\frac{\frac{1}{m}M_{1}+\frac{m-1}{m}M_{n}}{\frac{1}{m}M_{1}+\frac{m-1}{2m}M_{n}}\leq 2-\frac{2}{2+(m-1)\alpha}
\end{eqnarray*}
where $\alpha=\frac{M_{n}}{M_{1}}$. Since $\alpha\leq \frac{2}{n+1}$~\cite{eastman1964}, we obtain
\begin{eqnarray*}
2-\frac{2}{2+(m-1)\alpha}\leq 2-\frac{n+1}{n+m} \leq 2-\frac{n+1}{2n} \leq \frac{3}{2}.
\end{eqnarray*}

Therefore, we have
\begin{eqnarray}\label{ineq:08}
\frac{\frac{1}{m}M_{1}+\frac{m-1}{m}M_{n}}{\frac{1}{m}M_{1}+\frac{m-1}{2m}M_{n}}\leq \frac{3}{2}.
\end{eqnarray}
We also have
\begin{eqnarray}\label{ineq:09}
\frac{M_{\Lambda}}{M_{\Lambda^*}} \leq \frac{\phi_{max}}{\phi_{min}}.
\end{eqnarray}
By combining inequalities~(\ref{ineq:08}) and (\ref{ineq:09}), we obtain
\begin{eqnarray*}\label{ineq:10}
\frac{WCT_{LRF}}{WCT_{OPT}} \leq \max\left\{\frac{3}{2},\frac{\phi_{max}}{\phi_{min}}\right\}.
\end{eqnarray*}
\end{proof}

\subsection{The CosMOS Algorithm and Proof of Theorem~\ref{thm:thm3}}
Algorithm~\ref{Alg_ONLRF} gives the pseudocode of the CosMOS algorithm~\cite{Zhang2025} described in Section~\ref{sec:LRF}. Lines~\ref{Alg_ONLRF:1}--\ref{Alg_ONLRF:3} initialize the variables. Line~\ref{Alg_ONLRF:4} applies an offline algorithm $A$, such as LRF, to the remaining tasks and crowd workers to determine the tasks of crowd worker $i$. Line~\ref{Alg_ONLRF:5} processes these tasks in Smith-ratio order, and lines~\ref{Alg_ONLRF:6}--\ref{Alg_ONLRF:7} remove them from the set of unscheduled tasks and remove crowd worker $i$ from $\mathcal{M}$.
\begin{algorithm}[t]
\caption{The CosMOS Algorithm}
\label{Alg_ONLRF}
\begin{algorithmic}[1]
    \REQUIRE Set of crowd workers $\mathcal{M}$, set of tasks $\mathcal{J}$ \label{Alg_ONLRF:1}
    \ENSURE Final schedule $\Lambda_{CosMOS} = \{\mathcal{J}_1, \mathcal{J}_2, \ldots, \mathcal{J}_m\}$

    \WHILE{requester meets crowd worker $i \in \mathcal{M}$ on time $t_{i}$}
        \STATE $\mathcal{J}_i \leftarrow \emptyset, \quad EW_i \leftarrow \phi_{i}$
        \FOR{each $k \in \mathcal{M} \setminus \{i\}$}
            \STATE $\mathcal{J}_k \leftarrow \emptyset, \quad EW_k \leftarrow 2\phi_{k}-t_{i}$ \label{Alg_ONLRF:2}
        \ENDFOR \label{Alg_ONLRF:3}
				\STATE Apply offline algorithm $A$ to allocate the set of tasks $\mathcal{J}_{i}$ to crowd worker $i\in \mathcal{M}$. \label{Alg_ONLRF:4}
        \STATE Assign $\mathcal{J}_i$ to crowd worker $i\in \mathcal{M}$ and schedule using Smith ratio order \label{Alg_ONLRF:5}
        \STATE $\mathcal{J} \leftarrow \mathcal{J} \setminus \mathcal{J}_i$  \label{Alg_ONLRF:6}
        \STATE $\mathcal{M} \leftarrow \mathcal{M} \setminus \{i\}$ \label{Alg_ONLRF:7}
    \ENDWHILE

    \RETURN $\Lambda_{CosMOS} = \{\mathcal{J}_1, \mathcal{J}_2, \ldots, \mathcal{J}_m\}$
\end{algorithmic}
\end{algorithm}
Zhang \textit{et al.}~\cite{Zhang2025} showed that the online decisions of CosMOS never increase the total weighted completion time.
\begin{thm}[{[\citen{Zhang2025}]}]\label{thm:thm4}
In the CosMOS algorithm, the total weighted completion time (WCT) does not increase from one decision step to the next. Consequently, $WCT_{CosMOS} \leq WCT_{A}$, where $WCT_{A}$ is the WCT of the schedule computed by the offline algorithm $A$ at the beginning.
\end{thm}

The proof of Theorem~\ref{thm:thm3} relies on Theorem~\ref{thm:thm4}.
\begin{restated}{Theorem}{thm:thm3}Let $\Lambda_{OPT} = \{\mathcal{J}_1^*, \mathcal{J}_2^*, \dots, \mathcal{J}_m^* \}$ be an optimal schedule computed by an omniscient scheduler that knows in advance all meeting times between the requester and the crowd workers, and let $WCT_{OPT}$ be its cost. If the offline algorithm $A$ has approximation ratio $\alpha$, then
\begin{eqnarray*}
\frac{WCT_{CosMOS}}{WCT_{OPT}} & \leq & \alpha\left(1 + \frac{\Phi_{max} \sum_{j=1}^{n}w_{j}}{\sum_{j=1}^{n}w_{j}(\Phi_{min}+p_{j})}\right) \\
                               & \leq & \alpha\left(1 + \frac{\Phi_{max}}{\Phi_{min}+p_{min}}\right)
\end{eqnarray*}
where $p_{min}$ is the minimum required service time.
\end{restated}
\begin{proof}[Proof of Theorem~\ref{thm:thm3}]
The proof follows the approach of~\cite{Zhang2025, chen2025approx}. For a schedule $\Lambda=\left\{\mathcal{J}_{1}, \mathcal{J}_{2}, \ldots, \mathcal{J}_{m}\right\}$, the realized total weighted completion time is
\begin{eqnarray*}\label{ineq:12}
WCT_{\Lambda}=\sum_{i=1}^{m}\sum_{j\in \mathcal{J}_{i}} w_{j}(t_{i}+t'_{i}+T_{j}+p_{j}),
\end{eqnarray*}
where $t_{i}$ and $t'_{i}$ are the realized inter-meeting times between the requester and crowd worker $i$ for task distribution and for result feedback, respectively, and $T_{j}$ is the total service time of the tasks processed before $j \in \mathcal{J}_{i}$.

We also have
\begin{eqnarray*}\label{ineq:13}
WCT_{OPT}=\sum_{i=1}^{m}\sum_{j\in \mathcal{J}_{i}^{*}} w_{j}(t_{i}+t'_{i}+T_{j}+p_{j}).
\end{eqnarray*}

Let $WCT'$ be the expected cost of the schedule $\Lambda_{OPT}$ in the offline model, in which the realized delays are replaced by their expectations. Then
\begin{eqnarray*}\label{ineq:14}
WCT'& = & \sum_{i=1}^{m}\sum_{j\in \mathcal{J}_{i}^{*}} w_{j}\left(\Phi_{i}+T_{j}+p_{j}\right) \\
    & = & WCT_{OPT} +\sum_{i=1}^{m}\sum_{j\in \mathcal{J}_{i}^{*}} w_{j}\left(\Phi_{i}-t_{i}-t'_{i}\right) \\
		& \leq  & WCT_{OPT} +\sum_{i=1}^{m}\sum_{j\in \mathcal{J}_{i}^{*}} w_{j}\Phi_{i} \\
		& \leq  & WCT_{OPT} +\Phi_{max}\sum_{j=1}^{n}w_{j}.
\end{eqnarray*}

By Theorem~\ref{thm:thm4} and since $A$ is an $\alpha$-approximation in the offline model,
\begin{eqnarray*}\label{ineq:15}
WCT_{CosMOS} & \leq & WCT_{A} \\
             & \leq & \alpha \cdot WCT' \\
						 & \leq & \alpha \cdot WCT_{OPT}+ \alpha \Phi_{max} \sum_{j=1}^{n}w_{j}.
\end{eqnarray*}

Since $WCT_{OPT}\geq \sum_{j=1}^{n}w_{j}(\Phi_{min}+p_{j})\geq (\Phi_{min}+p_{min})\sum_{j=1}^{n}w_{j}$, we have
\begin{eqnarray*}\label{ineq:16}
\frac{WCT_{CosMOS}}{WCT_{OPT}} & \leq & \alpha\left(1 + \frac{\Phi_{max} \sum_{j=1}^{n}w_{j}}{\sum_{j=1}^{n}w_{j}(\Phi_{min}+p_{j})}\right) \\
& \leq & \alpha\left(1 + \frac{\Phi_{max} \sum_{j=1}^{n}w_{j}}{(\Phi_{min}+p_{min})\sum_{j=1}^{n}w_{j}}\right) \\
& = & \alpha\left(1 + \frac{\Phi_{max}}{\Phi_{min}+p_{min}}\right).
\end{eqnarray*}
\end{proof}

\section{Proof of Theorem~\ref{thm:thm44}}\label{supp:RTS}
\begin{restated}{Theorem}{thm:thm44}In the schedule produced by the RTS algorithm, the expected completion time of each task $j$ is at most $1.45\cdot C_{j}^{*}$, where $C_{j}^{*}=\sum_{i,s}x_{ijs}(\Phi_{i}+s+p_{ij})$ is the fractional completion time of $j$ in the optimal solution $x$ of the LP relaxation of~(\ref{LP:1}). Consequently, RTS is a randomized algorithm with expected approximation ratio 1.45.
\end{restated}
\begin{proof}[Proof of Theorem~\ref{thm:thm44}]
Let $x$ be the optimal solution of the LP relaxation of~(\ref{LP:1}) used by the RTS algorithm, and let $x_{ij}=\sum_{s}x_{ijs}$. We verify, step by step, that each ingredient of the analysis of Im and Li~\cite{Im2023} is unaffected by the contact delays $\Phi_{i}$.

\textit{Step 1 (LP objective and feasible region).} The contact delays appear only in the objective of~(\ref{LP:1}); constraints~(\ref{LP1:a})--(\ref{LP1:c}) are exactly those of the time-indexed LP in~\cite{Im2023}. Hence $x$ is a feasible solution of the LP of~\cite{Im2023}. The objective of~(\ref{LP:1}) decomposes as
\begin{eqnarray*}
\sum_{j}w_{j}\sum_{i,s}x_{ijs}(\Phi_{i}+s+p_{ij}) & = & \sum_{j}w_{j}\sum_{i}x_{ij}\Phi_{i} \\
 & & + \sum_{j}w_{j}\sum_{i,s}x_{ijs}(s+p_{ij}),
\end{eqnarray*}
where the second term is the objective of~\cite{Im2023} evaluated at $x$.

\textit{Step 2 (dominance, grouping and rounding).} The definition of dominance ($x_{ij}>1/2$), the random parameters $\rho$ and $\tau_{ij}$, the base windows and the groups, the assignment of anchor rectangles $R_{ijs}$ by the SNC rounding scheme (Theorem~\ref{thm:4}) and the priorities $\theta_{j}$ depend only on $x$ and on the service times $p_{ij}$; none of them depends on $\Phi_{i}$. Algorithm~\ref{Alg3} performs exactly Steps~1--3 of the rounding algorithm in~\cite[Section~3.2]{Im2023}. Therefore, the joint distribution of the assignment $\sigma$ and of the processing order on each crowd worker produced by the RTS algorithm is identical to that of the algorithm of~\cite{Im2023} applied to $x$. In particular, the marginal probabilities, the non-positive correlation and the strong negative correlation properties of Theorem~\ref{thm:4} hold unchanged.

\textit{Step 3 (completion-time decomposition).} Under the communication model of Section~\ref{sec:Preliminaries}, the expected distribution delay and the expected feedback delay of every task assigned to crowd worker $i$ are both $\phi_{i}$. Hence the contact delay is the same constant $\Phi_{i}$ for every task assigned to crowd worker $i$ and does not depend on the processing order. Let $C'_{j}$ be the completion time of task $j$ in the schedule of~\cite{Im2023}, in which crowd worker $i$ starts processing at time $0$. Then, conditioned on $\sigma(j)=i$, we have $C_{j}=\Phi_{i}+C'_{j}$.

\textit{Step 4 (bound).} The analysis in~\cite[Section~4]{Im2023} shows that, for any feasible solution $x$ of their LP, every task $j$ and every crowd worker $i$ with $x_{ij}>0$,
\begin{eqnarray*}
E[C'_{j}|\sigma(j)=i] \leq \frac{1.45}{x_{ij}}\sum_{s}x_{ijs}(s+p_{ij}),
\end{eqnarray*}
which is inequality~(4.7) of~\cite{Im2023}. It is proved by Lemma~4.1 of~\cite{Im2023} when $i$ dominates $j$ (the ratio is $1.5-\frac{1/2}{2}+0.2=1.45$), and by Section~4.1 of~\cite{Im2023} when $i$ does not dominate $j$: there, conditioned on the anchor rectangle $R_{ijs}$ and on $\tau_{ij}$, the expected total processing time of the tasks sequenced before $j$ is decomposed over configurations of rectangles on $i$ (Claim~4.3 and Lemma~4.4 of~\cite{Im2023}), which uses only constraint~(\ref{LP1:b}) and the correlation properties of Theorem~\ref{thm:4}, and deconditioning on $\tau_{ij}$ gives the ratio $1.45$ (Lemma~4.5 of~\cite{Im2023} and the subsequent bound). Both cases use only the LP constraints, which coincide with ours by Step~1, and the joint distribution of the assignment and of the priorities, which coincides with ours by Step~2.
Combining this inequality with Step~3 and using $\sum_{s}x_{ijs}/x_{ij}=1$ and $1\leq 1.45$, we obtain
\begin{eqnarray}\label{thm4:eq3}
 E[C_{j}|\sigma(j) = i] &\!\!=\!\!& \Phi_{i} + E[C'_{j}|\sigma(j)=i] \nonumber\\
 & \leq & \frac{1}{x_{ij}}\sum_{s}x_{ijs}\Phi_{i} + \frac{1.45}{x_{ij}}\sum_{s}x_{ijs}(s+p_{ij}) \nonumber\\
 & \leq & \frac{1.45}{x_{ij}}\sum_{s}x_{ijs} (\Phi_i + s + p_{ij}).
\end{eqnarray}
Since $\Pr[\sigma(j)=i]=x_{ij}$ (marginal probabilities), deconditioning~(\ref{thm4:eq3}) gives $E[C_{j}] \leq 1.45 \cdot \sum_{i,s}x_{ijs} (\Phi_i + s + p_{ij})=1.45\cdot C^{*}_{j}$. Summing over all tasks with weights $w_{j}$ shows that the expected objective is at most 1.45 times the optimal value of the LP relaxation of~(\ref{LP:1}), which is a lower bound on the optimal total weighted completion time.
\end{proof}

\section{Details of the DTS Algorithm}\label{supp:DTS}
This appendix gives the details of the \emph{Deterministic Time-indexed Scheduling} (DTS) algorithm (Algorithm~\ref{Alg4}) of Section~\ref{sec:Algorithm4}, which derandomizes RTS. In RTS, the rectangles are first partitioned into groups according to the base windows, tasks are assigned to groups by the SNC scheme, and then an anchor rectangle $R_{ijs}$ and an offset are chosen at random. A direct derandomization computes, for each task, the expected total weighted completion time for every crowd worker and every time point and selects the choice that minimizes it, which takes $O(mn^2T^2)$ time. However, once the assignment of tasks to crowd workers is fixed, processing the tasks of each crowd worker in Smith-ratio order is optimal~\cite{smith1956}. Hence the processing order is determined by the assignment, and the expected cost need not be evaluated at every time point.

Nevertheless, since the groups of a crowd worker are not disjoint, i.e., a task may belong to several groups, the conditional expectations must still be computed for each group, which takes $O(mn^2G)$ time, where $G$ is the number of groups per crowd worker. Moreover, the random choices that define the groups ($\rho$ and the offsets $\tau_{ij}$) must also be derandomized, which requires an exhaustive search over these parameters. Because this step is computationally expensive, we only describe the derandomization here; Section~\ref{sec:Algorithm5} presents a deterministic algorithm that does not use groups.

We use the groups $\mathcal{U}$, the mapping $g$, the values $y_{uj}$, and the assignment $\psi$ of Section~\ref{sec:Algorithm3}. Theorem~\ref{thm:4} allows us to bound the expected completion time of every task as follows.
Let $\psi^{-1}(u)=\{j \in \mathcal{J} : R_{ijs} \in u\}$ represent the set of all tasks in group $u$, and let $g^{-1}(i)=\{u \in \mathcal{U} : g(u)=i\}$ denote the set of all groups assigned to crowd worker $i$. Additionally, let $\mathcal{D}_{i}=\{j \in \mathcal{J} : \sum_{u \in g^{-1}(i)} y_{uj} > 1/2\}$ represent the set of all tasks dominated by crowd worker $i$. For each crowd worker $i$, we define a binary relation $\preceq_{i}$ on $\mathcal{J}$: for two tasks $j, j' \in \mathcal{J}$, $j' \preceq_{i} j$ if and only if$$\frac{w_{j'}}{p_{ij'}} > \frac{w_j}{p_{ij}} \text{ or } \left( \frac{w_{j'}}{p_{ij'}} = \frac{w_j}{p_{ij}} \text{ and } j' > j \right).$$ This relation follows Smith's rule: $j' \preceq_{i} j$ means that $j'$ is processed before $j$ on crowd worker $i$. Let $g(u)=i$. If crowd worker $i$ does not dominate task $j$, then
\begin{eqnarray*}
E[C_{j}|\psi(j)=u] & = & \Phi_{i}+p_{ij} \\
                     &   & +\sum_{u'\in g^{-1}(i)\setminus\left\{u\right\}}\sum_{j'\in \psi^{-1}(u'), j'\preceq_{i} j} y_{u'j'}p_{ij'} \\
										 &   & +\sum_{j'\in \psi^{-1}(u)\cap \mathcal{D}_{i}, j'\preceq_{i} j} y_{uj'}p_{ij'} \\
										 &   & +\sum_{j'\in \psi^{-1}(u)\setminus \mathcal{D}_{i}, j'\preceq_{i} j} (1-\eta)y_{uj'}p_{ij'}.
\end{eqnarray*}
Conversely, if crowd worker $i \in \mathcal{M}$ dominates task $j$, we have:
\begin{eqnarray*}
E[C_{j}|\psi(j)=u] & = & \Phi_{i}+p_{ij} \\
                     &   & +\sum_{u'\in g^{-1}(i)}\sum_{j'\in \psi^{-1}(u'), j'\preceq_{i} j} y_{u'j'}p_{ij'}.
\end{eqnarray*}
Accordingly, the expected completion time of task $j$ is given by:
\begin{eqnarray*}
E[C_{j}] & = & \sum_{u\in \mathcal{U}} y_{uj} E[C_{j}|\psi(j)=u].
\end{eqnarray*}

Let $\mathcal{P} \subseteq \mathcal{J}$ be the set of tasks that have already been assigned. For each task $j$, the binary variable $z_{ij}$ indicates whether $j$ has been assigned to crowd worker $i$, and $\bar{z}_{j} = 1 - \sum_{i}z_{ij}$ indicates that $j$ is still unassigned. With this notation, the conditional expected completion times are as follows.

\textbf{Case 1: $j \notin \mathcal{P}$ (Unassigned tasks)} If crowd worker $i \in \mathcal{M}$ does not dominate task $j$, we have:
\begin{eqnarray*}
E[C_{j}|\psi(j)=u] & = & \Phi_{i}+p_{ij}+\sum_{j'\in \mathcal{J}, j'\preceq_{i} j} z_{ij'} p_{ij'}  \\
                     &   & +\sum_{u'\in g^{-1}(i)\setminus\left\{u\right\}}\sum_{j'\in \psi^{-1}(u'), j'\preceq_{i} j} \bar{z}_{j'} y_{u'j'}p_{ij'} \\
										 &   & +\sum_{j'\in \psi^{-1}(u)\cap \mathcal{D}_{i}, j'\preceq_{i} j} \bar{z}_{j'}y_{uj'}p_{ij'} \\
										 &   & +\sum_{j'\in \psi^{-1}(u)\setminus \mathcal{D}_{i}, j'\preceq_{i} j} (1-\eta)\bar{z}_{j'}y_{uj'}p_{ij'}.
\end{eqnarray*}
Conversely, if crowd worker $i \in \mathcal{M}$ dominates task $j$, then:
\begin{eqnarray*}
E[C_{j}|\psi(j)=u] & = & \Phi_{i}+p_{ij}+\sum_{j'\in \mathcal{J}, j'\preceq_{i} j} z_{ij'} p_{ij'} \\
                     &   & +\sum_{u'\in g^{-1}(i)}\sum_{j'\in \psi^{-1}(u'), j'\preceq_{i} j} \bar{z}_{j'}y_{u'j'}p_{ij'}.
\end{eqnarray*}
Accordingly, for any $j \notin \mathcal{P}$, the overall expected completion time of task $j$ is given by:
\begin{eqnarray*}
E[C_{j}] & = & \sum_{u\in \mathcal{U}} y_{uj} E[C_{j}|\psi(j)=u].
\end{eqnarray*}

\textbf{Case 2: $j \in \mathcal{P}$ (Assigned tasks)}
If crowd worker $i \in \mathcal{M}$ does not dominate task $j$ and task $j$ has already been assigned to crowd worker $i$, its expected completion time is:
\begin{eqnarray*}
E[C_{j}] & = & \Phi_{i}+p_{ij}+\sum_{j'\in \mathcal{J}, j'\preceq_{i} j} z_{ij'} p_{ij'}  \\
                     &   & +\sum_{u'\in g^{-1}(i)\setminus\left\{u\right\}}\sum_{j'\in \psi^{-1}(u'), j'\preceq_{i} j} \bar{z}_{j'} y_{u'j'}p_{ij'} \\
										 &   & +\sum_{j'\in \psi^{-1}(u)\cap \mathcal{D}_{i}, j'\preceq_{i} j} \bar{z}_{j'}y_{uj'}p_{ij'} \\
										 &   & +\sum_{j'\in \psi^{-1}(u)\setminus \mathcal{D}_{i}, j'\preceq_{i} j} (1-\eta)\bar{z}_{j'}y_{uj'}p_{ij'}.
\end{eqnarray*}
Conversely, if crowd worker $i \in \mathcal{M}$ dominates task $j$ and task $j$ has already been assigned to crowd worker $i$, then:
\begin{eqnarray*}
E[C_{j}] & = & \Phi_{i}+p_{ij}+\sum_{j'\in \mathcal{J}, j'\preceq_{i} j} z_{ij'} p_{ij'} \\
                     &   & +\sum_{u'\in g^{-1}(i)}\sum_{j'\in \psi^{-1}(u'), j'\preceq_{i} j} \bar{z}_{j'}y_{u'j'}p_{ij'}.
\end{eqnarray*}


DTS uses the method of conditional expectations. The expected completion time of an unassigned task is a convex combination of its conditional expected completion times on the different crowd workers, and once the task is assigned, it becomes the conditional expectation for the chosen crowd worker. Line~\ref{alg4-1} first derandomizes the grouping so as to minimize $E[\sum_{j}w_{j} C_{j}]$. Lines~\ref{alg4-2}--\ref{alg4-3} then assign the tasks one by one to the crowd worker that minimizes the conditional expectation of the objective, and line~\ref{alg4-4} processes the tasks of each crowd worker in non-increasing order of the Smith ratio $w/p$.

\begin{algorithm}
\caption{The DTS Algorithm}
\label{Alg4}
\begin{algorithmic}[1]
    \STATE Solve the linear program (\ref{LP:1}) to obtain $y$.
    \STATE Derandomize the grouping method to minimize $E[\sum_{j} w_j C_j]$. \label{alg4-1}
    \STATE Initialize $\mathcal{P} \leftarrow \emptyset$, $z \leftarrow 0$, and $\mathcal{J}_i \leftarrow \emptyset$ for all $i \in \mathcal{M}$. \label{alg4-2}

    \FOR{each task $j \in \mathcal{J}$}
        \STATE Select crowd worker $i^* \leftarrow \arg\min_{i} E[\sum_{q} w_q C_q \mid \mathcal{P} \cup \{j\}, z_{ij}=1]$.
        \STATE Update $\mathcal{P} \leftarrow \mathcal{P} \cup \{j\}$, $z_{i^*j} \leftarrow 1$, and $\mathcal{J}_{i^*} \leftarrow \mathcal{J}_{i^*} \cup \{j\}$.
    \ENDFOR \label{alg4-3}

    \FOR{each crowd worker $i \in \mathcal{M}$}
        \STATE Schedule tasks in $\mathcal{J}_i$ in non-increasing order of the Smith ratio ($w/p$). \label{alg4-4}
    \ENDFOR
\end{algorithmic}
\end{algorithm}

\section{Proofs for Section~\ref{sec:Algorithm5}}\label{supp:EDTS}
\begin{restated}{Theorem}{thm:EDTS}Let $x$ be an optimal solution of the LP relaxation of~(\ref{LP:1}) and $LP^{*}$ its objective value. The EDTS algorithm outputs a schedule whose total weighted completion time is at most $1.5\cdot LP^{*}\leq 1.5\cdot WCT_{OPT}$.
\end{restated}
\begin{proof}[Proof of Theorem~\ref{thm:EDTS}]
Consider the following randomized procedure $\mathcal{R}$: each task $j$ is independently assigned to crowd worker $i$ with probability $y_{ij}=\sum_{s}x_{ijs}$, and the tasks assigned to each crowd worker are processed in Smith-ratio order. Let $\mathcal{R}'$ be the independent-rounding procedure of Section~\ref{sec:Algorithm3}: each task $j$ independently selects a rectangle $R_{ijs}$ with probability $x_{ijs}$ and the tasks on each crowd worker are processed in non-decreasing order of $\theta_{j}=s+\tau_{j}$. The assignment of tasks to crowd workers has the same distribution under $\mathcal{R}$ and $\mathcal{R}'$, because in both procedures task $j$ is assigned to crowd worker $i$ independently with probability $\sum_{s}x_{ijs}$.

Fix any realization of the assignment. On each crowd worker $i$, the completion time of every task equals the constant contact delay $\Phi_{i}$ plus its completion time in a single-machine schedule starting at time $0$. Since Smith's rule is optimal for minimizing the total weighted completion time on a single machine~\cite{smith1956}, and adding the constant $\Phi_{i}\sum_{j\in\mathcal{J}_{i}}w_{j}$ does not change the optimal order, the Smith-ratio order on crowd worker $i$ is no worse than the $\theta$-order of $\mathcal{R}'$. Taking expectations, $E_{\mathcal{R}}[\sum_{j}w_{j}C_{j}]\leq E_{\mathcal{R}'}[\sum_{j}w_{j}C_{j}]$. By the analysis of Section~\ref{sec:Algorithm3}, $E_{\mathcal{R}'}[C_{j}]\leq\sum_{i,s}x_{ijs}(\Phi_{i}+s+1.5p_{ij})\leq 1.5\sum_{i,s}x_{ijs}(\Phi_{i}+s+p_{ij})$. Hence $E_{\mathcal{R}}[\sum_{j}w_{j}C_{j}]\leq 1.5\cdot LP^{*}$.

The conditional expectations computed in Cases~1 and~2 are exactly the expected completion times under $\mathcal{R}$ given the partial assignment $z$. When EDTS processes task $j$, the current conditional expectation of the objective is the convex combination $\sum_{i}y_{ij}E[\sum_{q}w_{q}C_{q}\mid \mathcal{P}\cup\{j\},z_{ij}=1]$, so the minimizing choice $i^{*}$ does not increase it. By induction over the $n$ tasks, the final (deterministic) objective value is at most the initial expectation $E_{\mathcal{R}}[\sum_{j}w_{j}C_{j}]\leq 1.5\cdot LP^{*}$. Finally, line~\ref{alg5-4} processes the tasks on each crowd worker in Smith-ratio order, which is the order assumed by the conditional expectations, and $LP^{*}\leq WCT_{OPT}$ because the LP relaxation of~(\ref{LP:1}) is a relaxation of the problem.
\end{proof}

\begin{restated}{Lemma}{lem:LPI}The optimal value $LP_{I}^{*}$ of LP~(\ref{LP:Interval}) is a lower bound on the optimal total weighted completion time $WCT_{OPT}$.
\end{restated}
\begin{proof}[Proof of Lemma~\ref{lem:LPI}]
Consider an optimal schedule. Without loss of generality, each crowd worker processes its tasks without idle time, so that all processing is finished by time $\sum_{j}\max_{i}p_{ij}\leq\tau_{L}$ on the processing timeline (the contact delay $\Phi_{i}$ is added separately). For each task $j$ assigned to crowd worker $i$, let $p_{ij}x_{ij\ell}$ be the amount of processing of task $j$ during $I_{\ell}$, and set $x_{i'j\ell}=0$ for $i'\neq i$. Constraints~(\ref{LPI:a}) and~(\ref{LPI:b}) clearly hold. Let $S_{j}$ be the start time of task $j$ on the processing timeline; its mean busy time
\begin{eqnarray*}
M_{j}=\frac{1}{p_{ij}}\int_{S_{j}}^{S_{j}+p_{ij}}t\,dt=S_{j}+\frac{p_{ij}}{2}
\end{eqnarray*}
is the average time at which it is processed. Splitting the integral over the intervals $I_{\ell}=(\tau_{\ell-1},\tau_{\ell}]$ and using $t>\tau_{\ell-1}$ for every $t\in I_{\ell}$, we obtain
\begin{eqnarray*}
M_{j} & = & \sum_{\ell}\frac{1}{p_{ij}}\int_{[S_{j},S_{j}+p_{ij}]\cap I_{\ell}}t\,dt\\
& \geq & \sum_{\ell}\frac{1}{p_{ij}}\cdot p_{ij}x_{ij\ell}\cdot\tau_{\ell-1}=\sum_{\ell}x_{ij\ell}\tau_{\ell-1},
\end{eqnarray*}
since the length of $[S_{j},S_{j}+p_{ij}]\cap I_{\ell}$ is the amount $p_{ij}x_{ij\ell}$ of task $j$ processed during $I_{\ell}$. Therefore, $C_{j}=\Phi_{i}+S_{j}+p_{ij}=\Phi_{i}+M_{j}+p_{ij}/2$ satisfies~(\ref{LPI:c}), and $C_{j}\geq\Phi_{i}+p_{ij}$ gives~(\ref{LPI:d}). Hence the schedule yields a feasible LP solution with the same objective value.
\end{proof}

\begin{restated}{Theorem}{thm:EDTS-interval}For any fixed $\epsilon>0$, the polynomial-time version of EDTS is a deterministic $(1.5+\epsilon/2)$-approximation algorithm.
\end{restated}
\begin{proof}[Proof of Theorem~\ref{thm:EDTS-interval}]
Consider the following randomized procedure $\mathcal{R}'_{I}$: each task $j$ independently selects a pair $(i,\ell)$ with probability $x_{ij\ell}$, draws a time $t_{j}$ uniformly from $I_{\ell}$, and is assigned to crowd worker $i$; on each crowd worker, the tasks are processed without idle time in non-decreasing order of $t_{j}$. Condition on task $j$ selecting $(i,\ell)$ and $t_{j}=t$. A task $j'\neq j$ is processed before $j$ on crowd worker $i$ only if it selects some $(i,\ell')$ with $\ell'<\ell$, or selects $(i,\ell)$ with $t_{j'}<t$; the latter happens with probability $x_{ij'\ell}(t-\tau_{\ell-1})/(\tau_{\ell}-\tau_{\ell-1})$. By constraint~(\ref{LPI:b}), the expected total processing time of these tasks is at most $\sum_{\ell'<\ell}(\tau_{\ell'}-\tau_{\ell'-1})+(t-\tau_{\ell-1})=t$. Hence $E[C_{j}\mid (i,\ell),t_{j}=t]\leq\Phi_{i}+t+p_{ij}$, and taking the expectation over $t$,
\begin{eqnarray*}
E[C_{j}\mid (i,\ell)]\leq\Phi_{i}+p_{ij}+\frac{\tau_{\ell-1}+\tau_{\ell}}{2}.
\end{eqnarray*}
For $\ell\geq 1$, $\frac{\tau_{\ell-1}+\tau_{\ell}}{2}=(1+\frac{\epsilon}{2})\tau_{\ell-1}$; for $\ell=0$, it equals $\frac{\delta}{2}\leq\frac{\epsilon}{4}p_{ij}$. In both cases,
\begin{eqnarray*}
E[C_{j}\mid (i,\ell)] & \leq & \Phi_{i}+\Big(1+\frac{\epsilon}{2}\Big)\tau_{\ell-1}+\Big(1+\frac{\epsilon}{4}\Big)p_{ij}\\
& \leq & \Big(1+\frac{\epsilon}{2}\Big)\Big(\Phi_{i}+\tau_{\ell-1}+\frac{p_{ij}}{2}\Big)+\frac{1}{2}\big(\Phi_{i}+p_{ij}\big).
\end{eqnarray*}
Multiplying by $x_{ij\ell}$, summing over $(i,\ell)$, and using~(\ref{LPI:c}) and~(\ref{LPI:d}), we obtain $E[C_{j}]\leq(1.5+\epsilon/2)C_{j}$ for the LP values $C_{j}$, and thus $E[\sum_{j}w_{j}C_{j}]\leq(1.5+\epsilon/2)LP_{I}^{*}$.

The remaining argument is identical to the proof of Theorem~\ref{thm:EDTS}. Under $\mathcal{R}'_{I}$, task $j$ is assigned to crowd worker $i$ independently with probability $y_{ij}=\sum_{\ell}x_{ij\ell}$, which is the distribution used by EDTS. For every realized assignment, the Smith-ratio order on each crowd worker is no worse than the $t$-order, and the method of conditional expectations does not increase the expected objective. Therefore, EDTS returns a schedule of total weighted completion time at most $(1.5+\epsilon/2)LP_{I}^{*}\leq(1.5+\epsilon/2)WCT_{OPT}$ by Lemma~\ref{lem:LPI}.
\end{proof}

\section{Additional Experimental Results}\label{supp:exp}
This appendix reports the synthetic experiments that are omitted from Section~\ref{sec:Results}. The settings are those described in Section~\ref{sec:Results}.

Figure~\ref{fig:ratio4_box} presents the distribution of the WCTR at $n/m=50$ using box plots. EDTS and GREEDY-LS have the lowest medians (1.187 and 1.185) and the smallest interquartile ranges (0.009 and 0.010), followed by GREEDY (median 1.299). The medians of the LRF variants range from 1.473 (LRF-MIN) to 1.739 (LRF-MAX), and their interquartile ranges (0.042 to 0.064) are four to seven times larger than that of EDTS, indicating that EDTS is also considerably more robust across instances.

\begin{figure}[!htbp]
    \centering
        \includegraphics[width=3.8in]{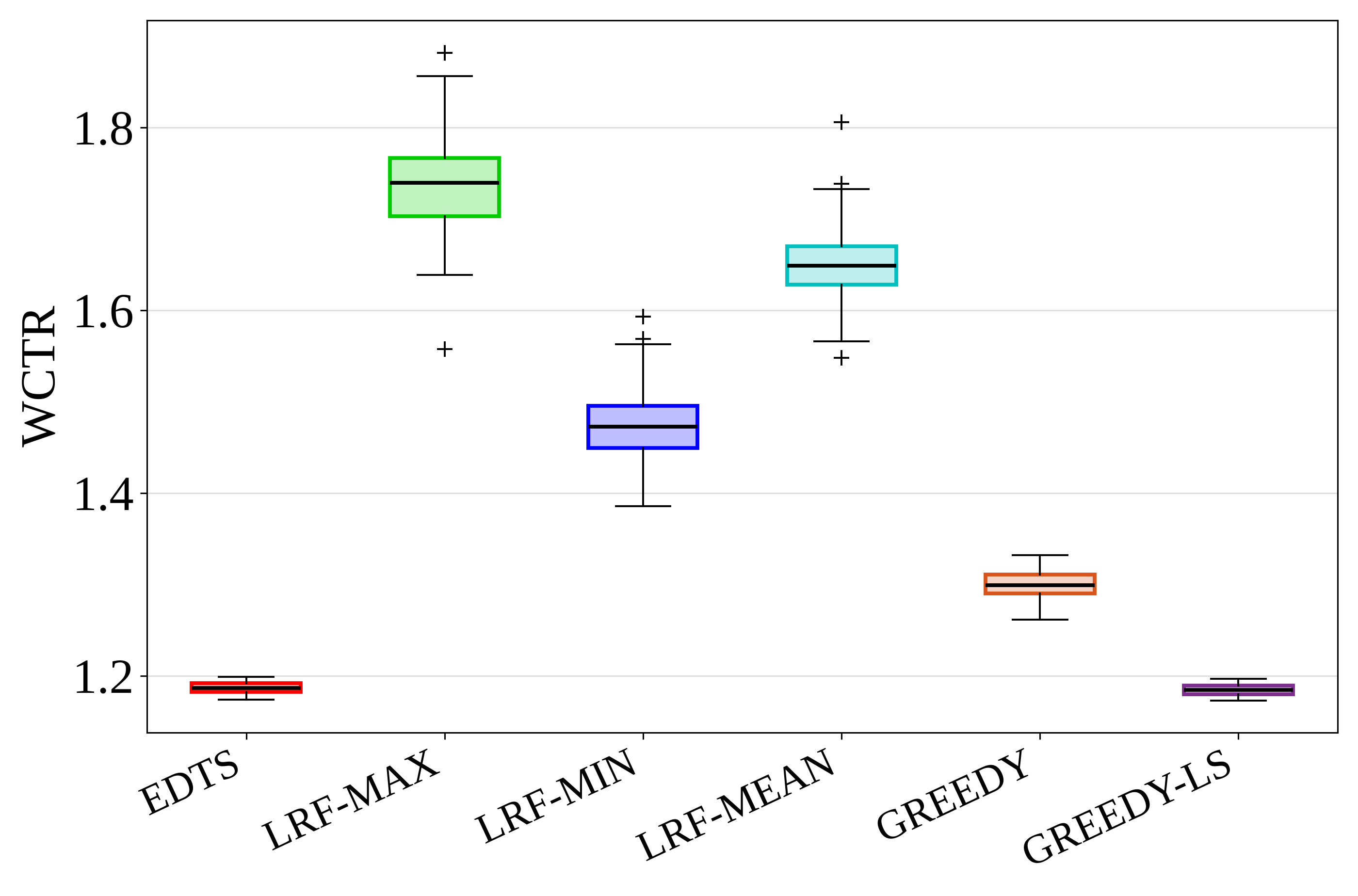}
    \caption{Box plot of the WCTR of the six algorithms at $n/m=50$.}
    \label{fig:ratio4_box}
\end{figure}Figure~\ref{fig:ratio4} demonstrates the impact of the mean of the base service time $\alpha$. As the mean increases from 5 to 50, the WCTR of EDTS increases moderately from 1.107 to 1.178, whereas that of LRF-MAX increases from 1.371 to about 1.67 and that of the best LRF variant from 1.318 to 1.430. EDTS and GREEDY-LS again almost coincide (difference at most 0.06\%), and both clearly outperform GREEDY (1.201 to 1.293).

\begin{figure}[!htbp]
    \centering
        \includegraphics[width=3.8in]{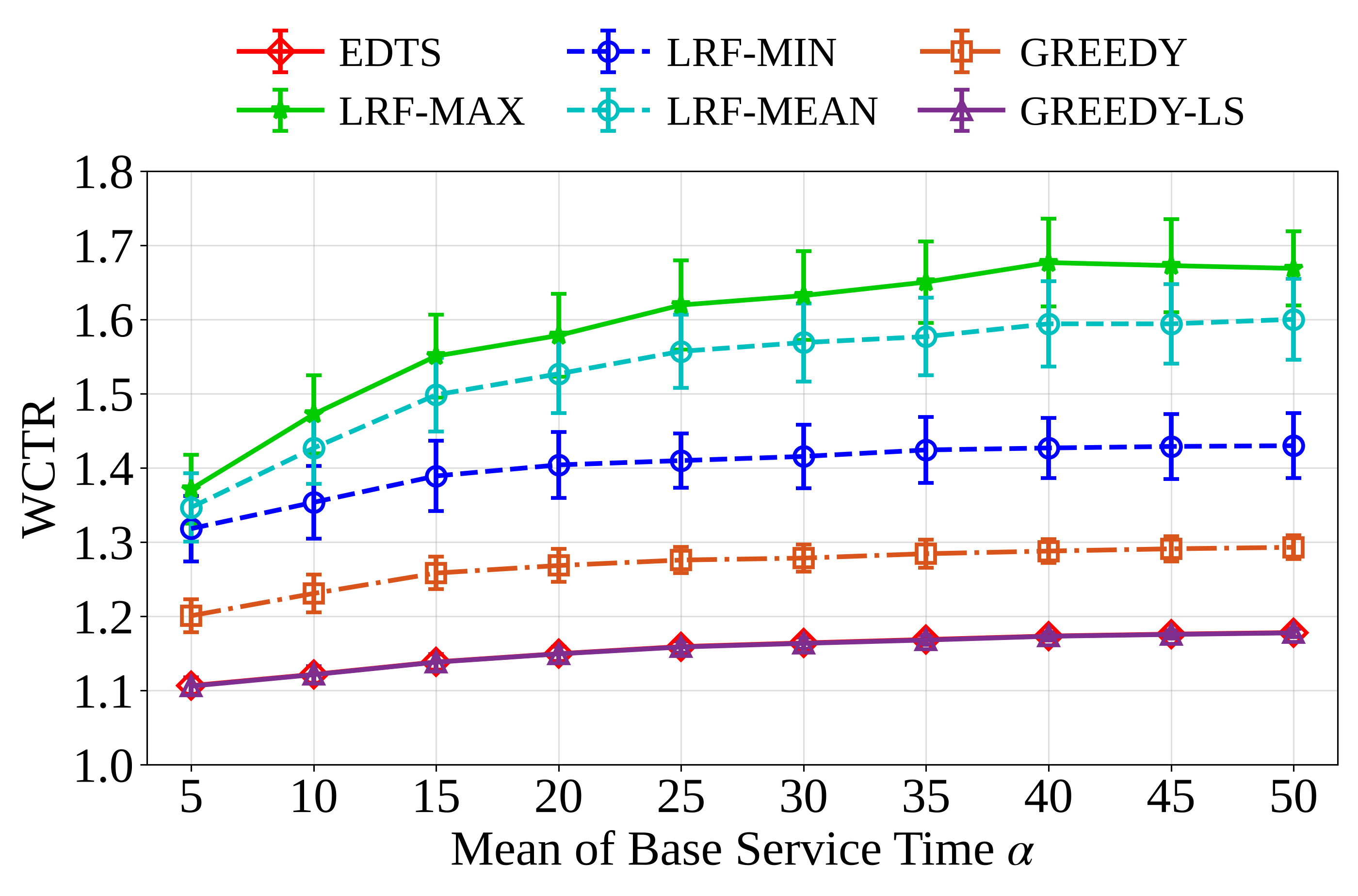}
    \caption{WCTR of the six algorithms versus the mean of the base service time $\alpha$.}
    \label{fig:ratio4}
\end{figure}

Figure~\ref{fig:ratio5} evaluates the sensitivity of the algorithms to the variance of the base service time $\alpha$, which ranges from 20 to 38. All curves are nearly horizontal: the WCTR of EDTS stays between 1.164 and 1.165, that of GREEDY between 1.278 and 1.281, and those of the LRF variants between 1.405 and 1.645. Hence, within this range, the variability of the task sizes has little influence on the relative performance of the algorithms.

\begin{figure}[!htbp]
    \centering
        \includegraphics[width=3.8in]{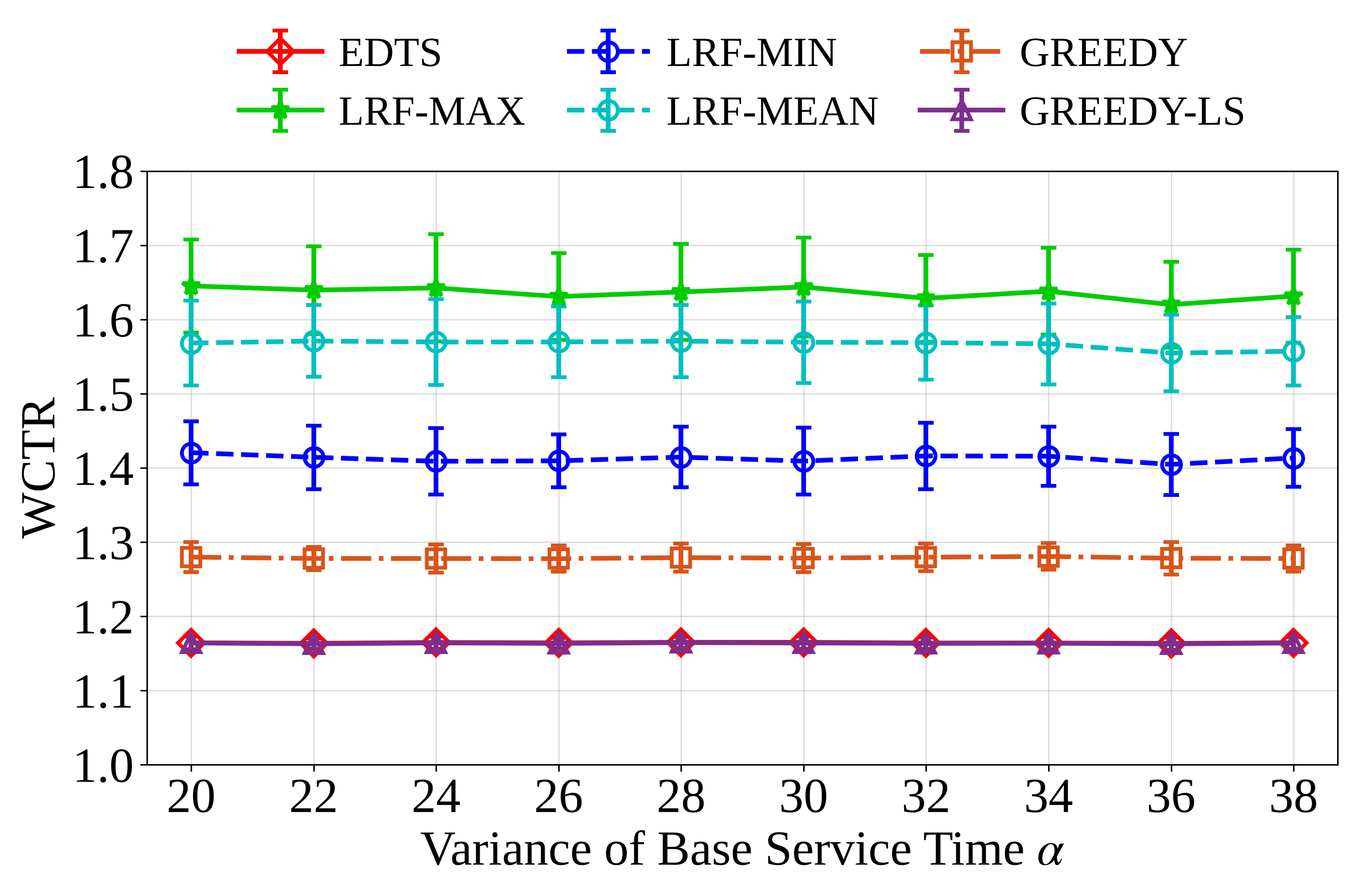}
    \caption{WCTR of the six algorithms versus the variance of the base service time $\alpha$.}
    \label{fig:ratio5}
\end{figure}

Figure~\ref{fig:ratio1} illustrates the impact of the ratio of the lower and upper limits of $\beta$, i.e., of the heterogeneity of the base capabilities of the crowd workers. The WCTR of EDTS increases only slightly, from 1.157 to 1.179, as the ratio increases from 0.1 to 1, and the curves of the LRF variants are nearly flat (1.389 to 1.659). EDTS outperforms all LRF variants by 15\% to 20\% over the whole range.

\begin{figure}[!htbp]
    \centering
        \includegraphics[width=3.8in]{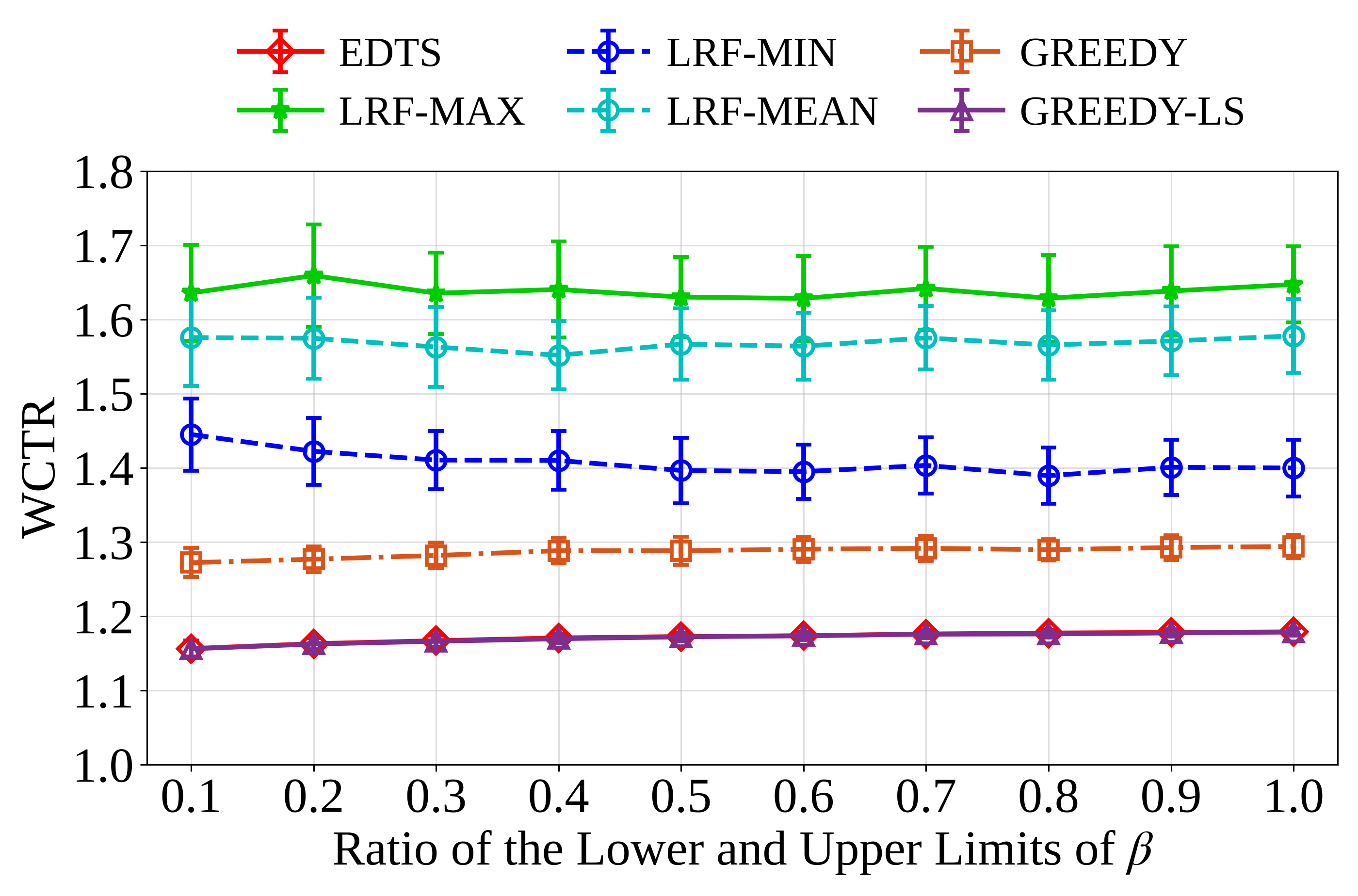}
    \caption{WCTR of the six algorithms versus the ratio of the lower and upper limits of $\beta$.}
    \label{fig:ratio1}
\end{figure}

\section{Notation}\label{supp:notation}
Table~\ref{tab:notations} summarizes the notation used in this paper.
\begin{table}[!htbp]
\caption{Notation and Terminology}
    \centering
        \begin{tabularx}{\textwidth}{||c|X||}
    \hline
		 $\mathcal{M}$      & The set of crowd workers. \\
		\hline
     $m$                & The number of crowd workers.          \\
    \hline
		 $\mathcal{J}$      & The set of indivisible crowdsourcing tasks. \\
		\hline
     $n$      & The number of tasks.          \\
    \hline
  	 $p_{ij}$ & The Required Service Time (RST) for task $j$ when processed by crowd worker $i$.         \\
    \hline
  	 $p_{j}$ & The RST of task $j$ when all crowd workers are identical ($p_{ij}=p_{j}$ for all $i$).         \\
    \hline
    $\phi_{i}$ & The expected inter-contact time between the requester $r$ and crowd worker $i$.        \\
    \hline
    $\Phi_{i}$ & $\Phi_{i}=2\phi_{i}$ is the total contact time between the requester $r$ and the crowd worker $i$.        \\
    \hline
     $EMT$      & The expected meeting time; the time between consecutive meetings of the requester $r$ and crowd worker $i$ is random, and its expectation is $\phi_{i}$.        \\
    \hline
     $EW_{i}$      & The expected workload of a crowd worker $i$.        \\
    \hline
     $C_{j}$ & The completion time of task $j$.         \\
    \hline
     $w_{j}$ & The weight of task $j$.         \\
    \hline
     $\phi_{max}, \phi_{min}$     & $\phi_{max}=\max_{i\in [1, m]} \phi_{i}$ is the maximum expected inter-contact time and $\phi_{min}=\min_{i\in [1, m]} \phi_{i}$ is the minimum expected inter-contact time. \\
    \hline
     $\Phi_{max}, \Phi_{min}$     & $\Phi_{max}=2\phi_{max}$ and $\Phi_{min}=2\phi_{min}$. \\
    \hline
     $T$     & The time horizon (number of unit time steps) of the time-indexed LP~(\ref{LP:1}). \\
    \hline
     $x_{ijs}$     & Variable of LP~(\ref{LP:1}); $x_{ijs}=1$ if task $j$ starts on crowd worker $i$ at time $s$. \\
    \hline
     $x_{ij}$, $y_{ij}$     & $x_{ij}=y_{ij}=\sum_{s}x_{ijs}$, the fraction of task $j$ assigned to crowd worker $i$ in the LP solution. \\
    \hline
     Dominated task     & Crowd worker $i$ dominates task $j$ if $x_{ij}>1/2$. \\
    \hline
     $L$, $\tau_{\ell}$, $I_{\ell}$     & The number of geometric intervals, the $\ell$-th time point, and the $\ell$-th interval $(\tau_{\ell-1},\tau_{\ell}]$ of the interval-indexed LP~(\ref{LP:Interval}). \\
    \hline
        \end{tabularx}
    \label{tab:notations}
\end{table}

\end{document}